\documentclass[journal]{IEEEtran}
\ifCLASSINFOpdf
\else
\fi
\ifCLASSOPTIONcompsoc
 \usepackage[caption=false,font=normalsize,labelfont=sf,textfont=sf]{subfig}
\else
 \usepackage[caption=false,font=footnotesize]{subfig}
\fi
\begin{document}
%
\title{Over-the-Air Federated Learning Over MIMO Channels: A Sparse-Coded Multiplexing Approach}
%
%
%

\author{Chenxi~Zhong,
        and~Xiaojun~Yuan,~\IEEEmembership{Senior~Member,~IEEE}
    \thanks{C. Zhong, and X. Yuan are with the National Key Laboratory of Science and Technology on Communications, University of Electronic Science and Technology of China, Chengdu, China (e-mail: cxzhong@std.uestc.edu.cn; xjyuan@uestc.edu.cn). The corresponding author is Xiaojun Yuan.}
}

\maketitle

\begin{abstract}
The communication bottleneck of over-the-air federated learning (OA-FL) lies in uploading the gradients of local learning models. 
In this paper, we study the reduction of the communication overhead in the gradients uploading by using the multiple-input multiple-output (MIMO) technique.
We propose a novel sparse-coded multiplexing (SCoM) approach that employs sparse-coding compression and MIMO multiplexing to balance the communication overhead and the learning performance of the FL model. 
We derive an upper bound on the learning performance loss of the SCoM-based MIMO OA-FL scheme by quantitatively characterizing the gradient aggregation error.
Based on the analysis results, we show that the optimal number of multiplexed data streams to minimize the upper bound on the FL learning performance loss is given by the minimum of the numbers of transmit and receive antennas. 
We then formulate an optimization problem for the design of precoding and post-processing matrices to minimize the gradient aggregation error. 
To solve this problem, we develop a low-complexity algorithm based on alternating optimization (AO) and alternating direction method of multipliers (ADMM), which effectively mitigates the impact of the gradient aggregation error. 
Numerical results demonstrate the superb performance of the proposed SCoM approach. 
\end{abstract}


\section{Introduction}
Sixth-generation (6G) wireless communications, as expected to support the connection density up to millions of wireless devices per square kilometer, will provide a solid foundation to fulfil the vision of ubiquitous intelligence \cite{zhou2020service}.
To develop a powerful intelligence model, it is necessary to exploit the diversity of data distributed over a large number of edge devices. A straightforward paradigm is to require edge devices to send local data to a central parameter server (PS) for training the model centrally. 
However, sending raw data to the PS requires a huge communication overhead and may expose user privacy. To overcome these drawbacks, federated learning (FL) is a promising substitute that allows edge devices to collaborate on training a machine learning (ML) model without sharing their local data with others \cite{mcmahan2017communication}.
Instead of uploading raw data, in an FL training round, each edge device sends its local gradient to the PS, and the PS aggregates the local gradients, updates and sends back the global model to the devices.

Gradient uploading is a critical bottleneck of deploying FL on a wireless network, since it is difficult to support the communication demands of massive edge devices with limited communication resources (e.g. time, bandwidth, and space). 
For example, the dimension of the recent ML models is extremely large, e.g., the ResNet152 has 60 million parameters \cite{he2016deep}, while the GPT-3 has 175 billion parameters \cite{brown2020language}. Yet, the available channel bandwidth is typically small due to the bandwidth and latency limitations, e.g., 1 LTE frame of 5MHz bandwidth and 10ms duration can carry only $50,000$ complex symbols. 
Fortunately, in FL, the PS does not need to know the local gradient of each device but the aggregated gradient, usually the mean of all the local gradients. 
Based on this property, over-the-air FL (OA-FL) is proposed in \cite{nazer2007computation, zhu2020broadband, amiri2020federated}, where edge devices share the same wireless resources to upload their local gradients. Thanks to the analog superposition of electromagnetic waves, the local gradients are aggregated over-the-air in the process of uploading. 
Compared with the traditional orthogonal multiple access (OMA) approaches \cite{yang2021energy,vu2020cell,vu2022joint}, OA-FL does not require more communication resources as the number of devices increases \cite{amiri2020federated}, which greatly alleviates the communication bottleneck of gradient uploading. 
Pioneering studies demonstrate that OA-FL also exhibits strong noise tolerance and significant latency improvement \cite{zhu2020broadband, liu2020reconfigurable}.

Due to the appealing features of OA-FL, much research effort has been devoted to the design of efficient OA-FL systems. For example, ref.~\cite{lin2017deep} pointed out that the local gradients can be sparsified, compressed, and quantized to reduce the communication overhead without causing substantial losses in accuracy. Ref.~\cite{amiri2020federated} proposed an efficient scheme, where the local gradients are sparsified and linear coding compressed before uploaded, and the aggregated gradient is recovered at the PS via compressed sensing methods. The authors in \cite{ma2022over} used partial-orthogonal compressing matrices and turbo compressed sensing (Turbo-CS) \cite{ma2014turbo}, achieving a lower complexity scheme of sparse coding. It has been shown in \cite{ma2022over} that sparse coding enables the OA-FL system to achieve a lower communication overhead and a faster convergence rate. 

The above schemes are all based on single-input single-output (SISO) systems.
Multiple-input multiple-output (MIMO) with array signal processing has been widely recognized as a powerful technique to enhance the system capacity.
MIMO multiplexing significantly reduces the number of channel uses by transmitting multiple spatial data streams in parallel through antenna arrays\cite{spencer2004zero}. 
However, MIMO multiplexing causes the inter-stream interference which corrupts the aggregated gradient and the test accuracy for OA-FL.
There have been some preliminary attempts to alleviate the impact of inter-stream interference by designing the precoding matrices at the devices and the post-processing matrix at the PS. 
For instance, ref.~\cite{zhu2018mimo} set the precoding matrix to the pseudo-inverse of channel matrix, and derived a closed-form equalizer as the post-processing matrix using differential geometry. 
The scheme proposed in \cite{chen2018over} also uses the pseudo-inverse of channel matrix as the precoding matrix, and computes the post-processing matrix based on the receive antenna selection.
However, both methods are based on the channel inversion, which may significantly amplify noise and hence exacerbate the gradients aggregation error, especially when some devices suffer from deep channel fading \cite{zhu2020broadband, liu2020reconfigurable, zhong2022over}. 

In this paper, we consider an over-the-air federated learning (OA-FL) network where the local gradients are uploaded over a MIMO multiple access (MAC) channel. 
The MIMO MAC channel comprises a central PS with multiple antennas and several edge devices with multiple antennas.
We propose a novel Sparse-Coded Multiplexing (SCoM) approach that integrates sparse coding with MIMO multiplexing in gradients uploading. 
Benefiting from two techniques, the SCoM achieves a strikingly better balance between the communication overhead and the learning performance. 
On one hand, sparse-coding utilizes the sparsity of the gradient to compress the gradient, reducing the communication overhead. 
On the other hand, MIMO multiplexing reduces the number of channel uses by transmitting multiple streams in parallel, and suppresses the gradients aggregation error through precoding and post-processing matrices.
The main contributions are summarized as follows.
\begin{itemize}
    \item We propose a novel SCoM approach for gradients uploading in MIMO OA-FL. 
    We derive an upper bound on the learning performance loss of the SCoM-based MIMO OA-FL scheme by quantitatively characterizing the gradient aggregation error. 
    \item Based on the analytical result, we formulate a joint precoding and post-processing matrices optimization problem for suppressing the gradient aggregation error. 
    We design a low-complexity algorithm that employs alternating optimization (AO) and alternating direction method of multipliers (ADMM) to jointly optimize the precoding and post-processing matrices.
    \item We derive the optimal number of multiplexed data streams for SCoM to balance the communication overhead and the gradient aggregation performance.
    More specifically, the optimal number of multiplexed data streams to minimize the upper bound of the learning performance loss is given by $\min\{N_{\mathrm{T}}, N_{\mathrm{R}}\}$, where $N_{\mathrm{T}}$ denotes the number of transmit antennas, and $N_{\mathrm{R}}$ denotes the number of receive antennas. 
\end{itemize}
Numerical results demonstrate that our proposed SCoM approach achieves the same test accuracy with much lower communication overhead than other existing approaches, which indicates the superior performance of the SCoM approach.

The rest of this paper is structured as follows. 
Section II introduces the FL model and the MIMO MAC channel. 
Section III presents the proposed SCoM approach. 
The analysis of the learning performance of the SCoM approach is presented in Section IV. 
In Section V, we present the optimization problem to minimize the gradient aggregation error, and propose a low-complexity algorithm to jointly optimize precoding and post-processing matrices. Section VI presents numerical results to evaluate the SCoM approach and Section VII concludes the paper. 

\textit{Notation}: 
$\mathbb{R}$ and $\mathbb{C}$ denote the sets of real and complex numbers, respectively. $\tr(\cdot)$, $\rank(\cdot)$, $(\cdot)^\dagger$, $(\cdot)^\mathrm{T}$, and $(\cdot)^\mathrm{H}$ are used to denote the trace, the rank, the conjugate, the transpose, and the conjugate transpose of the matrix, respectively. 
$[M]$ denotes the set $\{m | 1 \leq m \leq M\}$.
$\mathbf{s}(1:N)$ denotes a sub-vector of $\mathbf{s}$ that contains entries from index $1$ to index $N$.
The expectation operator is denoted by $\mathbb{E}[\cdot]$. 
We use $\mathbf{I}_N$, and $\mathbf{0}_{N\times M}$ to denote the identity matrix of size $N\times N$ and the zero matrix of size $N\times M$, respectively.
We use $\|\cdot\|_2$, and $\|\cdot\|_\mathrm{F}$ to denote the $l_2$-norm and the Frobenius norm, respectively. 
$\mathcal{CN}(\mu, \sigma)$ denotes the circularly-symmetric complex Gaussian (CSCG) distribution that has a mean of $\mu$ and a covariance of $\sigma$. 

\section{System Model}

\subsection{Federated Learning}
We start with the description of the FL task deployed on a wireless communication system, where the system consists of one central PS and $M$ edge devices. 
We assume that the training data of the FL task are all distributedly stored on the edge devices.
Let $\mathcal{A}_{m}$ denote the local dataset of the $m$-th device, and $Q_{m}$ denote the cardinality of $\mathcal{A}_{m}$. $Q = \sum_{m=1}^{M} Q_{m}$ is the total number of training data samples for the FL task. $\boldsymbol{\theta} \in \mathbb{R}^D$ is the model parameter vector with $D$ being the total length of the model parameter.
The target of the FL task is to minimize an empirical loss function $F(\boldsymbol{\theta})$ based on the local datasets $\{\mathcal{A}_{m}\}$, given by
\begin{equation}
    \min_{\boldsymbol{\theta}} \quad F(\boldsymbol{\theta}) = \frac{1}{Q} \sum\nolimits_{m=1}^{M} Q_{m} F_{m}(\boldsymbol{\theta}) = \frac{1}{Q} \sum\nolimits_{m=1}^{M} \sum\nolimits_{n=1}^{ Q_{m} } f\left(\boldsymbol{\theta}; \boldsymbol{\zeta}_{m,n} \right),
    \label{eq:FLTarget}
\end{equation}
where $F_{m}(\boldsymbol{\theta}) = \frac{1}{Q_{m}} \sum_{n=1}^{ Q_{m} } f\!(\boldsymbol{\theta}; \boldsymbol{\zeta}_{m,n})$ is the local loss function of device $m$, and $f\!(\boldsymbol{\theta}; \boldsymbol{\zeta}_{m,n})$ is the sample-wise loss function for the $n$-th training sample $\boldsymbol{\zeta}_{m,n}$ in $\mathcal{A}_{m}$. 

To minimize the empirical loss function $F(\boldsymbol{\theta})$ in \eqref{eq:FLTarget}, the FL training involves $T$ communication rounds between the edge devices and the PS for $F(\boldsymbol{\theta})$ to reach convergence. Specifically, each communication round $t \in [T]$ consists of four steps:
\begin{itemize}
    \item \textit{Global model download}:
    The PS sends the \textit{global model} $\boldsymbol{\theta}^{(t)}$ to each edge device.
    \item \textit{Local gradients computation}: The \textit{local gradient} $\mathbf{g}_{m}^{(t)}\in\mathbb{R}^D$ is computed by device $m$ based on their own data and the global model, given by
    \begin{equation}
        \mathbf{g}_{m}^{(t)} = \nabla F_{m}(\boldsymbol{\theta}^{(t)}) = \frac{1}{Q_{m}} \sum\nolimits_{n=1}^{ Q_{m}} \nabla f\left(\boldsymbol{\theta}^{(t)}; \boldsymbol{\zeta}_{m,n} \right).
        \label{eq:GradDef}
    \end{equation}
    \item \textit{Local gradients upload}: The edge devices send the local gradients to the PS through wireless channels.
    \item \textit{Global model aggregation}: The local gradients are aggregated as
    \begin{equation}
        \mathbf{g}^{(t)} = \frac{1}{Q} \sum\nolimits_{m = 1}^M Q_{m} \mathbf{g}_{m}^{(t)} = \sum\nolimits_{m = 1}^M q_{m} \mathbf{g}_{m}^{(t)},
        \label{eq:IdealGradient}
    \end{equation}
    where $\mathbf{g}^{(t)}$ denotes the \textit{aggregated gradient}, and $q_{m} = Q_{m}/Q, \forall m \in [M]$. The global model $\boldsymbol{\theta}^{(t+1)}$ is updated by 
    \begin{equation}
        \boldsymbol{\theta}^{(t+1)} = \boldsymbol{\theta}^{(t)}- \eta \mathbf{g}^{(t)},
        \label{eq:SGD}
    \end{equation}
    where $\eta \in \mathbb{R}$ denotes the learning rate.
\end{itemize}

\subsection{MIMO Channel Model}
    \begin{figure}[htbp]
    \centering
    \includegraphics[width=0.6\linewidth]{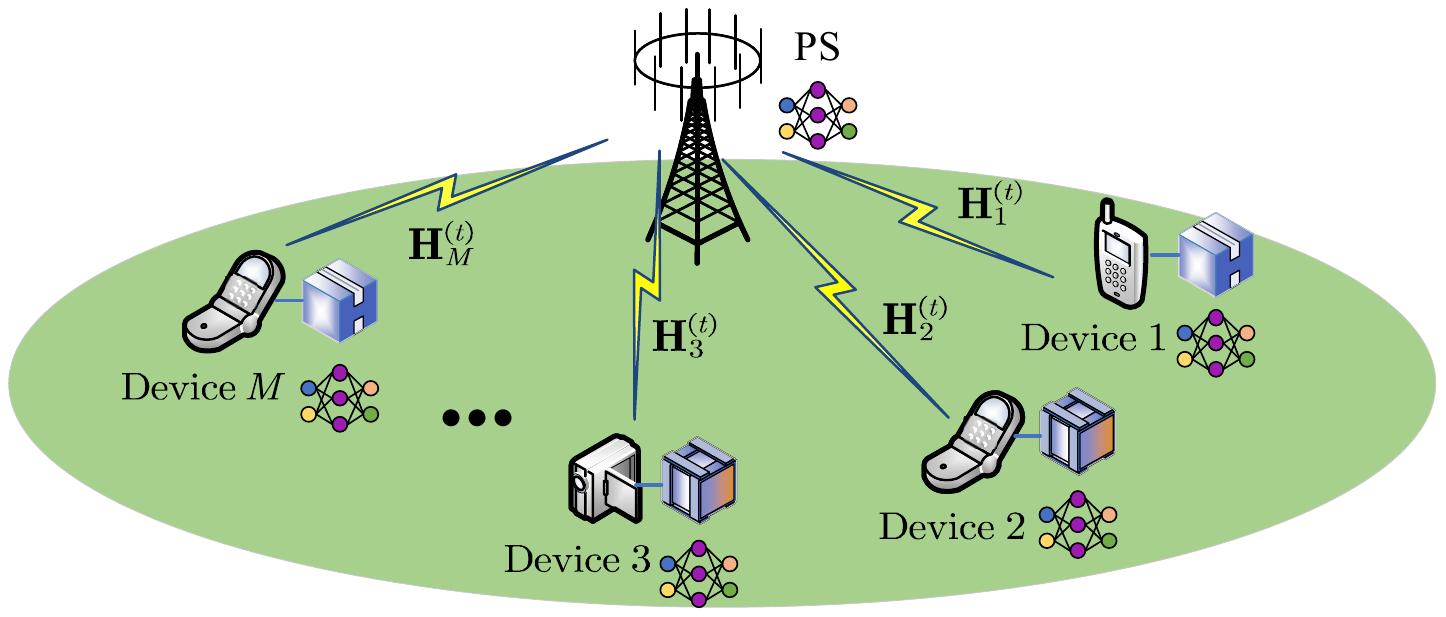}
    \caption{An illustration of the considered MIMO OA-FL system.}
    \label{fig:system_model}
\end{figure}
We now introduce the wireless multi-user multiple-input multiple-output (MIMO) channel for the above FL system. 
As depicted in Fig.~\ref{fig:system_model}, the considered MIMO OA-FL system consists of a PS with $N_{\mathrm{R}}$ antennas, and $M$ edge devices with each equipped with $N_{\mathrm{T}}$ antennas. 
As in previous studies in OA-FL \cite{amiri2020federated,sery2021over,cao2022transmission}, we make two assumptions: that the download of the global model is through error-free links\footnote{
In practice, the channel noise causes communication errors in the model download. This issue can be addressed by the schemes proposed by \cite{vu2020cell,vu2022joint}. This is beyond the scope of this paper and hence omitted here. }
and that the devices upload the local gradients to the PS synchronously\footnote{
The existing techniques in 4G Long Term Evolution, e.g., the timing advance (TA) mechanism, can achieve the synchronization of the gradient symbols among the edge devices \cite{3gpp.38.213}. }.
We now focus on the process of local gradients uploading. 
We consider a block-fading channel, where the channel state information (CSI) remains constant during the gradients uploading. 
Let $\mathbf{H}_{m}^{(t)} \in \mathbb{C}^{N_{\mathrm{R}} \times N_{\mathrm{T}}}$ denote the CSI matrix between the $m$-th device and the PS at the $t$-th round. We assume that the PS has perfect knowledge of the CSI of the wireless channels between the devices and the PS\footnote{
The approaches of CSI estimation over MIMO MAC channels can be referred to \cite{nguyen2013compressive,wen2014channel,vu2020cell,vu2022joint}.}. 
Thus, at the PS, the receive signal matrix from the above MIMO multiple access (MAC) channel is given by
\begin{equation}
    \mathbf{Y}^{(t)}
    = \sum\nolimits_{m=1}^M \mathbf{H}_{m}^{(t)} \mathbf{X}_{m}^{(t)} + \mathbf{N}^{(t)} \in \mathbb{C}^{N_{\mathrm{R}}\times K},
    \label{eq:receive_data}
\end{equation}
where $K$ denotes the number of channel uses at the $t$-th round; 
$\mathbf{X}_{m}^{(t)} \in \mathbb{C}^{N_{\mathrm{T}} \times K}$ denotes the transmit data matrix for the $m$-th device; 
and $\mathbf{N}^{(t)} \in \mathbb{C}^{N_{\mathrm{R}} \times K}$ is an additive white Gaussian noise (AWGN) matrix, with the entries independently drawn from $\mathcal{CN}(0,\sigma_{\mathrm{noise}})$.
Let $\mathbf{x}_{m,k}$ denote the $k$-th column of $\mathbf{X}_{m}^{(t)}$. Here, we consider the following transmit power constraint:
\begin{equation}
    \mathbb{E}[\|\mathbf{x}_{m,k}^{(t)}\|_2^2] \leq P_0, \forall k \in [K],
    \label{eq:trans_power_constr1}
\end{equation}
where $P_0$ is the transmit power budget.

What remains are to map the local gradients $\{\mathbf{g}_{m}^{(t)}\}$ to the transmit matrices $\{\mathbf{X}_{m}^{(t)}\}$ at the edge devices, and to recover the aggregated gradient from the receive signal matrix $\mathbf{Y}^{(t)}$. 
These issues are discussed in detail in the next section.

\section{Proposed Sparse-Coded Multiplexing Approach}
With the development of deep learning, the size of model is increasing. To upload the large number of FL local gradients over the aforementioned MIMO channel, the key challenge is the heavy communication burden.
Although transmitting the data streams in parallel with MIMO multiplexing efficiently reduces the communication overhead, MIMO multiplexing also causes the interference between the data streams, resulting in a loss of FL learning performance.

To address these challenges, we propose a novel transmission scheme, i.e., the Sparse-Coded Multiplexing (SCoM) approach, for the above MIMO OA-FL system, as shown in Fig.~\ref{fig:flow_figure}. 
The SCoM approach employs two techniques: sparse-coding and MIMO multiplexing. 
On one hand, sparse-coding utilizes the sparsity of the gradient to compress the gradient, reducing the communication overhead. Meanwhile, sparse-coding leverages the compression matrix to encode the data streams, thereby reducing the correlation between data streams and suppressing inter-stream interference.
On the other hand, MIMO multiplexing reduces the number of channel uses by transmitting multi-stream data through antenna arrays, and suppresses inter-stream interference through precoding and post-processing matrices.
The details of the SCoM approach are given as follows.
\begin{figure}[htbp]
\centering
\includegraphics[width=\linewidth]{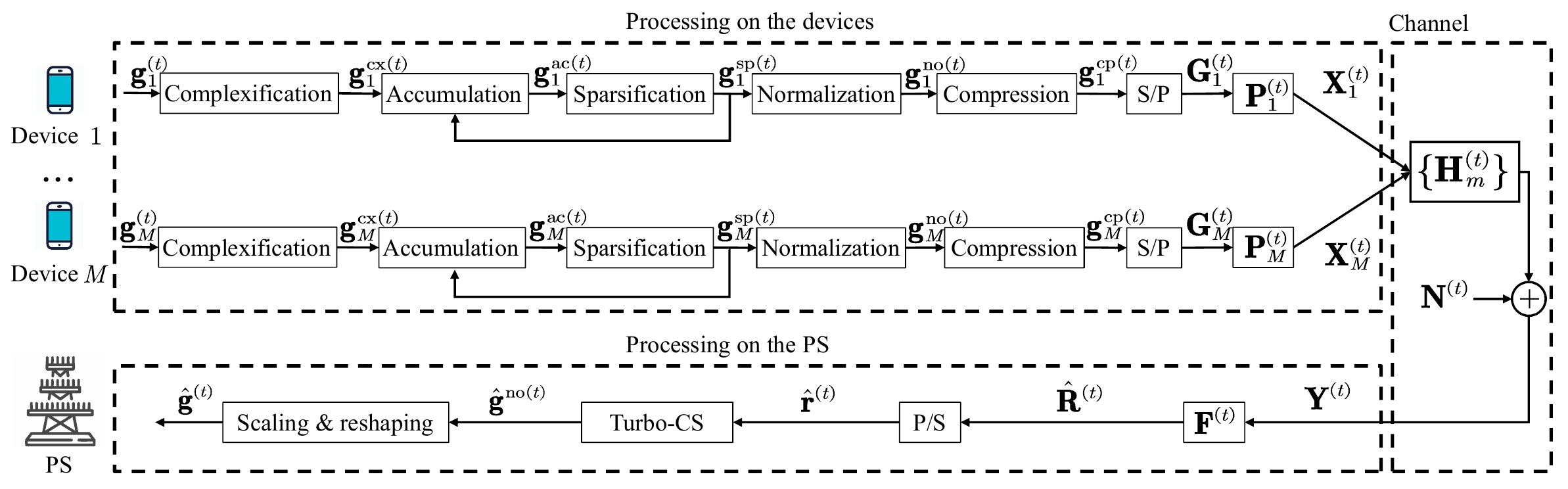}
\caption{A flow diagram of local gradients uploading at round $t$.}
\label{fig:flow_figure}
\end{figure}

\subsection{Processing on Devices}
To support the local gradients uploading in the MIMO OA-FL system, the pre-processing operations are first conducted on the edge devices, including the gradient sparsification \cite{amiri2020federated} and the gradient compression\cite{ma2014turbo}. 

To be specific, for the $m$-th edge device, the local gradient $\mathbf{g}_{m}^{(t)}$ at the $t$-th round is first complexified to fully utilize the spectral efficiency of complex channels. The complexified gradient is denoted by 
\begin{equation}
    \mathbf{g}_{m}^{\mathrm{cx}(t)} = \Re\{\mathbf{g}_{m}^{\mathrm{cx}(t)}\} + j \Im\{\mathbf{g}_{m}^{\mathrm{cx}(t)}\} \in \mathbb{C}^{D/2},
    \label{eq:Complexification}
\end{equation}
where $j = \sqrt{-1}$.
Then, the accumulated gradient is obtained by
\begin{equation}
    \mathbf{g}_{m}^{\mathrm{ac}(t)} = \mathbf{g}_{m}^{\mathrm{cx}(t)} + \boldsymbol{\Delta}_{m}^{(t)}  \in \mathbb{C}^{D/2},
    \label{eq:Accumulation}
\end{equation}
where $\boldsymbol{\Delta}_{m}^{(t)}$ denotes the error accumulation vector of the $m$-th device at the $t$-th round with $\Delta_{m}^{(0)}$ initialized as $\mathbf{0}$.
Having calculated the accumulated gradient, in the sparsification, the $m$-th device obtains the sparsified gradient via 
\begin{equation}
    \mathbf{g}_{m}^{\mathrm{sp}(t)} = \mathrm{sp}(\mathbf{g}_{m}^{\mathrm{ac}(t)}, \lambda) \in \mathbb{C}^{D/2},
    \label{eq:g_loc_sparsity}
\end{equation}
where $\lambda \in [0,1]$ denotes the sparsity ratio. The operator $\mathrm{sp}(\cdot)$ retains the $\lambda D / 2$ entries of $\mathbf{g}_{m}^{\mathrm{ac}(t)}$ with the largest absolute value magnitude, and sets the remaining $(1-\lambda) D / 2$ entries to $0$.
The error accumulation vector $\boldsymbol{\Delta}_{m}^{(t+1)}$ is updated via
\begin{equation}
   \boldsymbol{\Delta}_{m}^{(t+1)} = \mathbf{g}_{m}^{\mathrm{ac}(t)} - \mathbf{g}_{m}^{\mathrm{sp}(t)}.
   \label{eq:Delta_update}
\end{equation}
Then $\mathbf{g}_{m}^{\mathrm{sp}(t)}$ is normalized to $\mathbf{g}_{m}^{\mathrm{no}(t)}$ by \begin{equation}
    \mathbf{g}_{m}^{\mathrm{no}(t)} = (\mathbf{g}_{m}^{\mathrm{sp}(t)} \odot \mathbf{s}) / \sqrt{\sigma_{m}^{(t)}} \in \mathbb{C}^{D/2},
    \label{eq:g_normalize}
\end{equation}
where $\mathbf{s} \in \mathbb{C}^{D/2}$ is a random flipping vector, with each entry of $\mathbf{s}$ being independent and identically distributed (i.i.d.) drawn from $\{-1,+1\}$; and $\sigma_{m}^{(t)} = \frac{1}{D/2} \sum_{d=1}^{D/2} |g_{m}^{\mathrm{sp}(t)}[d]|^2 $ denotes the variance of $\{g_{m}^{\mathrm{sp}(t)}[d]\}_{d=1}^{D/2}$, with $g_{m}^{\mathrm{sp}(t)}[d]$ being the $d$-th entry of $\mathbf{g}_{m}^{\mathrm{sp}(t)}$. 
From \eqref{eq:g_normalize}, the entries of $\mathbf{g}_{m}^{\mathrm{no}(t)}$ have zero-mean and unit variance. 

Each device $m$ then compresses $\mathbf{g}_{m}^{\mathrm{no}(t)}$ into a low-dimensional vector via a common compressing matrix. Specifically, the compressed gradient is given by
\begin{equation}
    \mathbf{g}_{m}^{\mathrm{cp}(t)} = \mathbf{A} \mathbf{g}_{m}^{\mathrm{no}(t)} \in \mathbb{C}^{C},
    \label{eq:g_compress}
\end{equation}
where $\mathbf{A} \in \mathbb{C}^{C \times (D/2) }$ denotes the common compressing matrix\footnote{
We assume that both the compression matrix $\mathbf{A}$ and the flipping vector $\mathbf{s}$ keep invariant throughout the FL training process, and are shared among the devices prior to the FL training.}, 
$C$ denotes the length of the compressed gradient, and $\kappa = \frac{C}{D/2}$ denotes the compression ratio. 
Inspired by \cite{ma2014turbo, ma2022over}, we employ a partial DFT matrix as the compressing matrix, given by $\mathbf{A} = \mathbf{S} \boldsymbol{\Xi}$. $\mathbf{S} \in \mathbb{R}^{C \times D/2}$ is a selection matrix consisting of $C$ randomly selected and reordered rows of the $ D/2 \times  D/2$ identity matrix $\mathbf{I}_{D/2}$; and $\boldsymbol{\Xi} \in \mathbb{R}^{D/2 \times D/2}$ is a unitary DFT matrix, where the $(d,d^\prime)$-th entry of $\boldsymbol{\Xi}$ is given by $\frac{1}{\sqrt{D/2}} \exp\left(-2\pi j \frac{(d - 1)(d^{\prime} - 1)}{D/2}\right)$.
We note that the partial DFT matrix has lower computational complexity and better performance, compared with other types of compressed matrix such as i.i.d. Gaussian matrix \cite{ma2015performance}.

To transmit the data with multiple streams, device $m$ then reshapes the compressed gradient $\mathbf{g}_{m}^{\mathrm{cp}(t)}$ into $\mathbf{G}_{m}^{(t)} \in \mathbb{C}^{N_{\mathrm{S}}\times K}$ as
\begin{equation}
    \mathbf{G}_{m}^{(t)} 
    = [ \mathbf{g}_{m,1}^{\mathrm{cp}(t)}, \cdots, 
    \mathbf{g}_{m,N_{\mathrm{S}}}^{\mathrm{cp}(t)}]^T, 
    \label{eq:G_def}
\end{equation}
where $N_{\mathrm{S}}$ denotes the number of data streams, $\mathbf{g}_{m,n}^{\mathrm{cp}(t)} = \mathbf{g}_{m}^{\mathrm{cp}(t)}((n-1)K+1: nK), n \in [N_{\mathrm{S}}] $.
Naturally, the number of channel uses satisfies the following equation: 
\begin{equation}
    K = \frac{C}{N_{\mathrm{S}}} = \frac{\kappa D/2}{N_{\mathrm{S}}}.
    \label{eq:channel_uses}
\end{equation}

We are now ready to describe the design of the transmit matrix $\mathbf{X}_{m}^{(t)}$. 
The transmit matrix $\mathbf{X}_{m}^{(t)}$ is given by 
$\mathbf{X}_{m}^{(t)} = \mathbf{P}_{m}^{(t)}  \mathbf{G}_{m}^{(t)}$, 
where $\mathbf{P}_{m}^{(t)} \in \mathbb{C}^{N_{\mathrm{T}} \times N_{\mathrm{S}}}$ denotes the precoding matrix for the $m$-th device.
Let $\mathbf{g}_{m,k}^{(t)} \in \mathbb{C}^{N_{\mathrm{S}}} $ denote the $k$-th column of $\mathbf{G}_{m}^{(t)}$. We have $\mathbf{x}_{m,k}^{(t)} = \mathbf{P}_{m}^{(t)} \mathbf{g}_{m,k}^{(t)}, \forall k \in [K]$. 
From the transmit power constraint in \eqref{eq:trans_power_constr1}, we have 
\begin{equation}
    \mathbb{E}[\|\mathbf{P}_{m}^{(t)} \mathbf{g}_{m,k}^{(t)}\|_2^2] 
    \overset{(a)}{=} \|\mathbf{P}_{m}^{(t)} \|_\mathrm{F}^2 \leq P_0, \forall k \in [K],
\end{equation}
where the step (a) is due to the normalization in \eqref{eq:g_normalize}. 

\subsection{Processing on the PS}
We now describe the processing operations on the PS. 
In the following, we first transfer the recovery of the aggregated gradient to a compressed sensing problem, and then adopt the Turbo compressed sensing (Turbo-CS) algorithm \cite{ma2014turbo} to solve this problem.

At the PS, the receive signal matrix $\mathbf{Y}$ in \eqref{eq:receive_data} is first processed through the post processing matrix $\mathbf{F}^{(t)}$, and the post-processed matrix $\hat{\mathbf{R}}^{(t)}$ is given by
\begin{equation}
    \hat{\mathbf{R}}^{(t)} = \mathbf{F}^{(t)} \left( \sum\nolimits_{m=1}^M \mathbf{H}_{m}^{(t)} \mathbf{P}_{m}^{(t)} \mathbf{G}_{m}^{(t)} + \mathbf{N}^{(t)} \right) \in \mathbb{C}^{N_{\mathrm{S}}\times K}.
    \label{eq:post_processed_matrix}
\end{equation}
$\hat{\mathbf{R}}$ is an approximation to the compressed gradient aggregation matrix $\mathbf{R}^{(t)} = \sum\nolimits_{m=1}^M q_{m} \mathbf{G}_{m}$, and the residual error is given by
\begin{equation}
    \mathbf{W}^{(t)} = \hat{\mathbf{R}}^{(t)} - \mathbf{R}^{(t)}.
    \label{eq:error_receive_data}
\end{equation}
Based on \eqref{eq:error_receive_data}, the PS converts the post-processed matrix $\hat{\mathbf{R}}$ into an equivalent vector form:
\begin{equation}
    \hat{\mathbf{r}}^{(t)} = \vectorize(\hat{\mathbf{R}}^{(t)}{}^{\mathrm{T}}) 
    = \sum\nolimits_{m=1}^M q_{m} \vectorize \left( \mathbf{G}_{m}^{(t)\mathrm{T}} \right) + \vectorize \left( \mathbf{W}^{(t)}{}^{\mathrm{T}} \right)
    \overset{(a)}{=} \mathbf{A} \mathbf{g}^{\mathrm{no}(t)} + \mathbf{w}^{(t)} \in \mathbb{C}^{C},
    \label{eq:compressed_problem}
\end{equation}
where step (a) is from \eqref{eq:g_compress} and \eqref{eq:G_def}, together with the definition of $\mathbf{g}^{\mathrm{no}(t)} = \sum_{m=1}^M q_{m} \mathbf{g}_{m}^{\mathrm{no}(t)}$ and $\mathbf{w}^{(t)} = \vectorize \left( \mathbf{W}^{(t)}{}^{\mathrm{T}} \right)$.
We assume that the entries of $\mathbf{g}^{\mathrm{no}(t)}$ are independently drawn from a Bernoulli Gaussian distribution:
\begin{equation}
    g_d^{\mathrm{no}(t)} \sim 
    \begin{cases}
    0,                    & \text{probability} = 1 - \lambda^{\prime(t)},\\
    \mathcal{CN}(0, \sigma_g^{(t)}),& \text{probability} = \lambda^{\prime(t)},
    \end{cases}, \forall d \in [D/2] 
    \label{eq:g_agg_sparsity}
\end{equation}
where $g_d^{\mathrm{no}(t)}$ is the $d$-th entry of $\mathbf{g}^{\mathrm{no}(t)}$, $\lambda^{\prime(t)}$ is the sparsity ratio of $\mathbf{g}^{\mathrm{no}(t)}$, which can be estimated by the Expectation-Maximization algorithm \cite{vila2013expectation}, and $\sigma_g^{(t)}$ is the variance of the nonzero entries in $\mathbf{g}^{\mathrm{no}(t)}$. 
Moreover, we assume that the entries of $\mathbf{w}^{(t)}$ are i.i.d. drawn from $\mathcal{CN}(0, \sigma_{w}^{(t)})$, where $\sigma_{w}^{(t)}$ represents the mean square error (MSE) of the compressed gradient aggregation, given by
\begin{equation}
    \sigma_{w}^{(t)} = \frac{1}{N_{\mathrm{S}}K}\mathbb{E} \|\hat{\mathbf{R}}^{(t)} - \mathbf{R}^{(t)} \|_\mathrm{F}^2 .
    \label{eq:error_mimo}
\end{equation}

The recovery of $\mathbf{g}^{\mathrm{no}(t)}$ from $\hat{\mathbf{r}}^{(t)}$ in \eqref{eq:compressed_problem} is a compressed sensing problem, where the local gradients $\{\mathbf{g}_{m}^{(t)}\}$ are compressed by the partial DFT matrix $\mathbf{A}$ on the edge devices.
From \cite{ma2014turbo, ma2022over}, we see that the Turbo-CS algorithm is the state-of-the-art to solve the compressed sensing problem with partial-orthogonal sensing matrices. Thus, we employ the Turbo-CS algorithm to recover $\mathbf{g}^{\mathrm{no}(t)}$:
\begin{equation}
    \hat{\mathbf{g}}^{\mathrm{no}(t)} = \text{Turbo-CS}(\hat{\mathbf{r}}^{(t)}) \in \mathbb{C}^{C/2}.
\end{equation}
The details of the Turbo-CS algorithm are presented in the next subsection.
Meanwhile, the performance of the Turbo-CS algorithm to recover the aggregated gradient is related to the variance of the noise $\mathbf{w}^{(t)}$, i.e., $\sigma_{w}^{(t)}$. A smaller $\sigma_{w}^{(t)}$ leads to a better recovery performance and a less loss of learning accuracy, which is discussed in Section IV.

\subsection{Turbo-CS Algorithm}
As shown in Fig.~\ref{fig:Turbo-CS}, Turbo-CS conducts the iteration between modules A and B until convergence. Module A is a linear minimum mean-squared error (LMMSE) estimator handling the linear constraint in \eqref{eq:compressed_problem}, and module B is a minimum mean-squared error (MMSE) denoiser exploiting the sparsity in \eqref{eq:g_agg_sparsity}. 
We next give the operations of Turbo-CS in detail. 

\begin{figure}[ht]
    \centering
    \includegraphics[width = 0.7\linewidth]{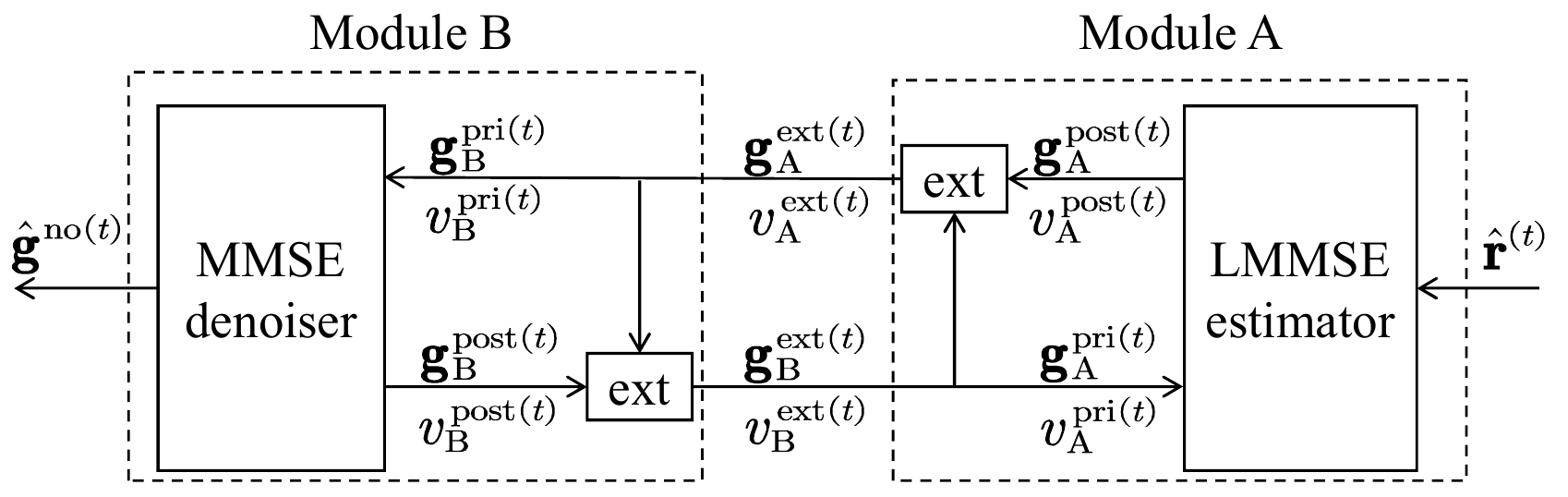}
    \caption{An illustration of the Turbo-CS algorithm.}
    \label{fig:Turbo-CS}
\end{figure}

The iterative process begins with module A. The inputs of the LMMSE estimator in module A are the \textit{a priori} mean $\mathbf{g}_{\mathrm{A}}^{\mathrm{pri}(t)} \in \mathbb{C}^{D/2}$, the \textit{a priori} covariance $v_{\mathrm{A}}^{\mathrm{pri}(t)}$, and the observed vector $\hat{\mathbf{r}}^{(t)}$. 
With given $\mathbf{g}_{\mathrm{A}}^{\mathrm{pri}(t)}$, $v_{\mathrm{A}}^{\mathrm{pri}(t)}$, and $\hat{\mathbf{r}}^{(t)}$, the \textit{a posteriori} mean $\mathbf{g}_{\mathrm{A}}^{\mathrm{post}(t)}$ and the \textit{a posteriori} covariance $v_{\mathrm{A}}^{\mathrm{post}(t)}$ of the LMMSE estimator are given by \cite{kay1993fundamentals}
\begin{subequations}
    \label{eq:g_v_post_A}
    \begin{align}
        \mathbf{g}_{\mathrm{A}}^{\mathrm{post}(t)}
        & = \mathbf{g}_{\mathrm{A}}^{\mathrm{pri}(t)}
        + \frac{v_{\mathrm{A}}^{\mathrm{pri}(t)}}{ v_{\mathrm{A}}^{\mathrm{pri}(t)}  + \sigma_{w}^{(t)} } \mathbf{A}^{\mathrm{H}}
        (\hat{\mathbf{r}}^{(t)} - \mathbf{A} \mathbf{g}_{\mathrm{A}}^{\mathrm{pri}(t)}), \\
        v_{\mathrm{A}}^{\mathrm{post}(t)} 
        & = v_{\mathrm{A}}^{\mathrm{pri}(t)} - \kappa \cdot  \frac{( v_{\mathrm{A}}^{\mathrm{pri}(t)})^2 }{ v_{\mathrm{A}}^{\mathrm{pri}(t)}  + \sigma_{w}^{(t)} },
        \label{eq:v_post_A}
    \end{align}
\end{subequations}
where $\kappa$ is the compression ratio defined below \eqref{eq:g_compress}.
Then the extrinsic mean and variance of the LMMSE estimator are given by 
\begin{subequations}
    \begin{align}
        \mathbf{g}_{\mathrm{A}}^{\mathrm{ext}(t)} &= v_{\mathrm{A}}^{\mathrm{ext}(t)} \left ( \frac{\mathbf{g}_{\mathrm{A}}^{\mathrm{post}(t)}}{v_{\mathrm{A}}^{\mathrm{post}(t)}}-\frac{\mathbf{g}_{\mathrm{A}}^{\mathrm{pri}(t)}}{v_{\mathrm{A}}^{\mathrm{pri}(t)}} \right ),\\
        v_{\mathrm{A}}^{\mathrm{ext}(t)} &= \left(1/{v_{\mathrm{A}}^{\mathrm{post}(t)}}-1/{v_{\mathrm{A}}^{\mathrm{pri}(t)}} \right){}^{-1}.
        \label{eq:v_ext_A}
    \end{align}
\end{subequations}
The extrinsic messages $\{\mathbf{g}_{\mathrm{A}}^{\mathrm{ext}(t)}, v_{\mathrm{A}}^{\mathrm{ext}(t)}\}$ are used to update the \textit{a priori} mean $\mathbf{g}_{\mathrm{B}}^{\mathrm{pri}(t)}$ and the \textit{a priori} variance $v_{\mathrm{B}}^{\mathrm{pri}(t)}$ as $\mathbf{g}_{\mathrm{B}}^{\mathrm{pri}(t)} = \mathbf{g}_{\mathrm{A}}^{\mathrm{ext}(t)}$, and $v_{\mathrm{B}}^{\mathrm{pri}(t)} = v_{\mathrm{A}}^{\mathrm{ext}(t)}$. 
Both $\{\mathbf{g}_{\mathrm{B}}^{\mathrm{pri}(t)}, v_{\mathrm{B}}^{\mathrm{pri}(t)}\}$ are the inputs of the MMSE denoiser in module B.

In module B, we model the \textit{a priori} mean $\mathbf{g}_{\mathrm{B}}^{\mathrm{pri}(t)}$ as an observation of $\mathbf{g}^{\mathrm{no}(t)}$ corrupted by additive noise: $\mathbf{g}_{\mathrm{B}}^{\mathrm{pri}(t)}\!=\!\mathbf{g}^{\mathrm{no}(t)}\!+\!\boldsymbol{\delta}^{(t)}$, where $\delta_d^{(t)} \sim \mathcal{CN}(0, v_{\mathrm{B}}^{\mathrm{pri}(t)} )$ denotes the $d$-th entry of $\boldsymbol{\delta}^{(t)}$. The \textit{a posteriori} mean $\mathbf{g}_{\mathrm{B}}^{\mathrm{post}(t)}$ and the variance $v_{\mathrm{B}}^{\mathrm{post}(t)}$ of the MMSE denoiser are given by
\begin{subequations}
    \label{eq:g_v_post_B}
    \begin{align}
        \mathbf{g}_{\mathrm{B}}^{\mathrm{post}(t)} &= \mathbb{E}\left[\mathbf{g}^{\mathrm{no}(t)}|\mathbf{g}_{\mathrm{B}}^{\mathrm{pri}(t)}\right], \\
        v_{\mathrm{B}}^{\mathrm{post}(t)} &= \frac{1}{D/2}\sum\nolimits_{d=1}^{D/2} \operatorname{var}\left[g^{\mathrm{no}(t)}_{d}|g_{\mathrm{B},d}^{\mathrm{pri}(t)}\right], 
        \label{eq:v_post_B}
    \end{align}
\end{subequations}
where $\operatorname{var}[a|b] = \mathbb{E}\Big[\big|a - \mathbb{E}[a|b]\big|^2 \Big| b\Big]$, and $g_{\mathrm{B},d}^{\mathrm{pri}(t)}$ is the $d$-th entry of $\mathbf{g}_{\mathrm{B}}^{\mathrm{pri}(t)}$. 
Then the extrinsic messages of the MMSE denoiser are given by
\begin{subequations}
    \label{eq:g_v_ext_B}
    \begin{align}
        &\mathbf{g}_{\mathrm{B}}^{\mathrm{ext}(t)} = v_{\mathrm{B}}^{\mathrm{ext}(t)} 
        \left(\frac{\mathbf{g}_{\mathrm{B}}^{\mathrm{post}(t)}}{v_{\mathrm{B}}^{\mathrm{post}(t)}} - \frac{\mathbf{g}_{\mathrm{B}}^{\mathrm{pri}(t)}}{v_{\mathrm{B}}^{\mathrm{pri}(t)}} \right), \\
        & v_{\mathrm{B}}^{\mathrm{ext}(t)} = \left( 1/{v_{\mathrm{B}}^{\mathrm{post}(t)}} - 1/{v_{\mathrm{B}}^{\mathrm{pri}(t)}} \right){}^{-1}.
        \label{eq:v_ext_B}
    \end{align}
\end{subequations}
The extrinsic messages $\{\mathbf{g}_{\mathrm{B}}^{\mathrm{ext}(t)}, v_{\mathrm{B}}^{\mathrm{ext}(t)}\}$ are used to update the \textit{a priori} mean $\mathbf{g}_{\mathrm{A}}^{\mathrm{pri}(t)}$ and the \textit{a priori} variance $v_{\mathrm{A}}^{\mathrm{pri}(t)}$ as $\mathbf{g}_{\mathrm{A}}^{\mathrm{pri}(t)} = \mathbf{g}_{\mathrm{B}}^{\mathrm{ext}(t)}$, and $v_{\mathrm{A}}^{\mathrm{pri}(t)}=v_{\mathrm{B}}^{\mathrm{ext}(t)}$.
Both $\{\mathbf{g}_{\mathrm{A}}^{\mathrm{pri}(t)}, v_{\mathrm{A}}^{\mathrm{pri}(t)}\}$ are the inputs of the LMMSE estimator in module A.
At the end of the iterative process, the final estimate is based on the \textit{a posteriori} output of the module B, i.e., $\hat{\mathbf{g}}^{\mathrm{no}(t)} = \mathbf{g}_{\mathrm{B}}^{\mathrm{post}(t)}$.

Then, based on the output of Turbo-CS, the aggregated gradient $\hat{\mathbf{g}}^{(t)}$ is given by
\begin{equation}
    \hat{\mathbf{g}}^{(t)} = [\Re \{\sqrt{\sigma^{(t)}} (\hat{\mathbf{g}}^{\mathrm{no}(t)} \odot \mathbf{s}) \},  \Im \{\sqrt{\sigma^{(t)}} (\hat{\mathbf{g}}^{\mathrm{no}(t)} \odot \mathbf{s}) \}] \in \mathbb{R}^D
    \label{eq:g_scaling}
\end{equation}
where $\sigma^{(t)} = \frac{1}{M}\sum_{m=1}^M \sigma_{m}^{(t)}$. The PS then updates the model $\boldsymbol{\theta}^{(t+1)}$ by 
\begin{equation}
    \boldsymbol{\theta}^{(t+1)} = \boldsymbol{\theta}^{(t)} - \eta \hat{\mathbf{g}}^{(t)}.
    \label{eq:SGD_error}
\end{equation} 

To sum up, at each round $t$, the receive matrix $\hat{\mathbf{R}}^{(t)}$ is first reshaped into a vector form in \eqref{eq:compressed_problem}. Given the observed vector $\hat{\mathbf{r}}^{(t)}$ and the initial values $\{\mathbf{g}_{\mathrm{A}}^{\mathrm{pri}(t)} = \mathbf{0}, v_{\mathrm{A}}^{\mathrm{pri}(t)} = 1\}$, Turbo-CS iteratively calculates \eqref{eq:g_v_post_A}-\eqref{eq:g_v_ext_B} until a certain termination criterion is met. Finally, the output $\hat{\mathbf{g}}^{\mathrm{no}(t)}$ is scaled as $\hat{\mathbf{g}}^{(t)}$ for the model update in \eqref{eq:SGD_error}. 

\subsection{Overall Scheme}
The proposed SCoM approach to the local gradients uploading is summarized below, where lines 3-5, 11-17 are executed at the PS, and lines 6-10 are executed at the devices.   

\begin{algorithm}[htb]
\renewcommand{\thealgorithm}{}
\floatname{algorithm}{}
\caption{Proposed SCoM-Based MIMO OA-FL Scheme} 
\label{alg:FL_framework} 
\begin{algorithmic}[1] 
\REQUIRE $T$, $\{Q_{m}\}$, $N_{\mathrm{S}}$, $N_{\mathrm{R}}$, $N_{\mathrm{T}}$, $\kappa$, $\lambda$, $\mathbf{A}$ and $\mathbf{s}$. 
\STATE {\textbf{Initialization:}} $t = 0$, the global model $\boldsymbol{\theta}^{(0)}$.
\FOR{ $t \in [T]$ }
    \STATE {\textbf{PS does:}}
    \STATE Estimate the CSI matrices $\{\mathbf{H}^{(t)}\}$, and optimize $\mathbf{F}^{(t)}$ and $\{\mathbf{P}_{m}^{(t)}\}$;
    \STATE Send $\boldsymbol{\theta}^{(t)}$ and $\{\mathbf{P}_{m}^{(t)}\}$ to the edge devices through error free links;
    \STATE {\textbf{Each device $m$ does in parallel:}}
    \STATE Compute $\mathbf{g}_{m}^{(t)}$ based on the $\mathcal{A}_{m}$ and $\boldsymbol{\theta}^{(t)}$ via \eqref{eq:GradDef};
    \STATE Compute $\mathbf{G}_{m}^{(t)}$ via \eqref{eq:Complexification}-\eqref{eq:G_def};
    \STATE Update $\boldsymbol{\Delta}_{m}^{(t+1)}$ via \eqref{eq:Delta_update};
    \STATE Send $\mathbf{G}_{m}^{(t)}$ to the PS with the precoding matrix $\mathbf{P}_{m}^{(t)}$ over the MIMO channel in \eqref{eq:receive_data};
    \STATE {\textbf{PS does:}}
    \STATE Reshape $\hat{\mathbf{R}}^{(t)}$ into $\hat{\mathbf{r}}^{(t)}$ via \eqref{eq:compressed_problem};
    \REPEAT
    \STATE Update $\mathbf{g}_{\mathrm{B}}^{\mathrm{post}(t)}$ via \eqref{eq:g_v_post_A}-\eqref{eq:g_v_ext_B};
    \UNTIL{convergence}
    \STATE Compute $\hat{\mathbf{g}}^{(t)}$ with $\hat{\mathbf{g}}^{\mathrm{no}(t)}$ via \eqref{eq:g_scaling};
    \STATE Update $\boldsymbol{\theta}^{(t+1)}$ via \eqref{eq:SGD_error};
\ENDFOR 
\end{algorithmic}
\end{algorithm}

There are two remaining issues in the design of the SCoM approach.
The first issue is to determine the number of multiplexed spatial streams, i.e., $N_{\mathrm{S}}$.
Generally speaking, increasing $N_{\mathrm{S}}$ reduces the number of channel uses required in gradient uploading, and hence reduces the communication overhead, as seen in \eqref{eq:channel_uses}. However, the interference between multiplexed spatial streams deteriorates with the increase of $N_{\mathrm{S}}$, which exacerbates the performance of gradient aggregation. 
This implies that we need to appropriately choose the value of $N_{\mathrm{S}}$ to strike a balance between the communication overhead and the gradient aggregation performance.

The second issue is to minimize the compressed gradient aggregation MSE $\sigma_{w}^{(t)}$, as previously mentioned in Section III-B. 
As seen in \eqref{eq:post_processed_matrix} and \eqref{eq:error_mimo}, to minimize $\sigma_{w}^{(t)}$ needs to appropriately design the post-processing matrix $\mathbf{F}^{(t)}$ and the precoding matrices $\{\mathbf{P}_{m}^{(t)}\}$. 
Note that a straightforward approach to design $\{\mathbf{P}_{m}^{(t)}\}$ is channel inversion, where $\{\mathbf{P}_{m}^{(t)}\}$ and $\mathbf{F}^{(t)}$ are designed to guarantee the equivalent channel gain matrices $\{\mathbf{F}^{(t)} \mathbf{H}_{m}^{(t)} \mathbf{P}_{m}^{(t)}\}$ to be diagonal matrices. 
However, it is well known that channel inversion suffers from noise amplification, especially when some devices are in deep fading. 
This motivates us to design a more robust algorithm to optimize the post-processing matrix $\mathbf{F}^{(t)}$ and the precoding matrices $\{\mathbf{P}_{m}^{(t)}\}$. 

In what follows, based on the standard assumptions in stochastic optimization, we first derive an upper bound on the learning performance loss of the SCoM-based MIMO OA-FL scheme. 
We then obtain the optimal number of $N_{\mathrm{S}}$ to minimize the upper bound in Section-IV. Furthermore, we present the proposed low-complexity optimization algorithm to minimize $\sigma_{w}^{(t)}$ by optimizing $\mathbf{F}^{(t)}$ and $\{\mathbf{P}_{m}^{(t)}\}$ in Section-V. In numerical results, we demonstrate the outstanding performance of our proposed algorithm.

\section{Performance Analysis}
In this section, we analyze the performance of the SCoM approach. Based on some standard assumptions in stochastic optimization, we first characterize the gradient aggregation error of the SCoM-based MIMO OA-FL scheme, and analyze the impact of this error on the FL learning performance. Finally, we provide a theorem about the optimal number of multiplexed data streams for SCoM to balance the learning performance and the communication overhead. 

\subsection{Assumptions}
We first make the following standard assumptions on the loss function $F(\cdot)$ in stochastic optimization \cite{friedlander2012hybrid}.
\begin{assumption}
    \label{asp:f_twice_diff}
    $F(\cdot)$ is continuously differentiable, and twice-continuously differentiable. 
\end{assumption}
\begin{assumption}
    \label{asp:f_Lipschitz}
    The gradient $\nabla F(\cdot)$ is $\omega$-Lipschitz continuous: 
    \begin{equation}
        \| \nabla F(\boldsymbol{\theta}) - \nabla F(\boldsymbol{\theta}^\prime)\|_2 \leq \omega \|\boldsymbol{\theta} - \boldsymbol{\theta}^\prime \|_2, \forall \boldsymbol{\theta}, \boldsymbol{\theta}^\prime \in \mathbb{R}^D.
        \label{eq:Lipschitz}
    \end{equation} 
\end{assumption}
\begin{assumption}
    \label{asp:f_convex}
    $F(\cdot)$ is strongly $\mu$-convex:
    \begin{equation}
        F(\boldsymbol{\theta}) 
        \geq F(\boldsymbol{\theta}^\prime) + (\boldsymbol{\theta} - \boldsymbol{\theta}^\prime)^\mathrm{T} \nabla F(\boldsymbol{\theta}^\prime) + \frac{\mu}{2} \| \boldsymbol{\theta} - \boldsymbol{\theta}^\prime\|_2^2, \forall \boldsymbol{\theta}, \boldsymbol{\theta}^\prime \in \mathbb{R}^D.
    \end{equation}
\end{assumption}
\begin{assumption}
    \label{asp:f_bound}
    The $l_2$-norm of sample-wise gradient $\nabla f\left(\boldsymbol{\theta}; \boldsymbol{\zeta}_{m,n}\right)$ is bounded by
    \begin{align}
        \| \nabla f\left(\boldsymbol{\theta}; \boldsymbol{\zeta}_{m,n}\right) \|_2^2 \leq \chi_1 + \chi_2 \| \nabla F(\boldsymbol{\theta}^{(t)}) \|_2^2,
        \label{eq:f_bound}
    \end{align}
    for some constants $\chi_1 \geq 0$ and $\chi_2 \geq 0$.
\end{assumption}

\subsection{Learning Performance Analysis}
We first discuss the gradient aggregation error.
Let $\mathbf{e}^{(t)} = \nabla F (\boldsymbol{\theta}^{(t)}) - \hat{\mathbf{g}}^{(t)}$ denote the gradient aggregation error, which is divided into two parts: 
\begin{equation}
    \mathbf{e}^{(t)} = 
    \underbrace{\nabla F (\boldsymbol{\theta}^{(t)}) -\tilde{\mathbf{g}}^{\mathrm{sp}(t)}}_{\mathbf{e}_{\mathrm{sp}}^{(t)} } 
    + \underbrace{\tilde{\mathbf{g}}^{\mathrm{sp}(t)} - \hat{\mathbf{g}}^{(t)} }_{ \mathbf{e}_{\mathrm{com}}^{(t)} }, 
    \label{eq:ErrorTwoPart}
\end{equation}
where $\mathbf{e}_{\mathrm{sp}}^{(t)}$ denotes the error caused by the sparsification on the devices, and $\mathbf{e}_{\mathrm{com}}^{(t)}$ denotes the communication error induced by the wireless channel with inter-stream interference and noise, and the recovery error of Turbo-CS. $\tilde{\mathbf{g}}^{\mathrm{sp}(t)}$ denotes the aggregation of the local sparsified gradients, given by
\begin{equation}
    \tilde{\mathbf{g}}^{\mathrm{sp}(t)} = \sum\nolimits_{m=1}^M q_{m} \tilde{\mathbf{g}}_{m}^{\mathrm{sp}(t)} \in \mathbb{R}^D,
    \label{eq:IdealSpGradient}
\end{equation}
where $\tilde{\mathbf{g}}_{m}^{\mathrm{sp}(t)} = [\Re\{\mathbf{g}_{m}^{\mathrm{sp}(t)}\}^\mathrm{T}, \Im \{ \mathbf{g}_{m}^{\mathrm{sp}(t)}\}^\mathrm{T}]^\mathrm{T}$.

Leveraging the preceding analysis of $\mathbf{e}^{(t)}$, we establish an upper bound on the loss function $F(\cdot)$ under Assumptions~\ref{asp:f_twice_diff}-\ref{asp:f_bound}, as stated in the following lemma.
\begin{lemma}
    \label{Lemma:Error_up}
    Under Assumptions~\ref{asp:f_twice_diff}-\ref{asp:f_bound} on the loss function $F(\cdot)$, at the $t$-th round, with the learning rate $\eta = 1/\omega$, we have
    \begin{equation}
        \mathbb{E}[F(\boldsymbol{\theta}^{(t+1)})]\!\leq\! \mathbb{E}[F(\boldsymbol{\theta}^{(t)})]\!-\!\frac{1}{2 \omega} \mathbb{E}[\|\nabla F(\boldsymbol{\theta}^{(t)})\|_2^{2}] + \!\frac{1}{\omega} (\mathbb{E}[\|\mathbf{e}_{\mathrm{sp}}^{(t)}\|_2^2] + \mathbb{E}[\|\mathbf{e}_{\mathrm{com}}^{(t)}\|_2^2]),
        \label{eq:Lemma_upper_bound_two_parts}
    \end{equation}
    where $\omega$ denotes the parameter of Lipschitz continuity defined in \eqref{eq:Lipschitz}, and $\mathbb{E}[\cdot]$ denotes the expectation with respect to (w.r.t.) $\{ \mathbf{N}^{(\tau)}, \mathbf{G}_{m}^{(\tau)} | m \in [M], \tau \in [t+1] \} $.
\end{lemma}
\begin{proof}
    Based on Assumptions~\ref{asp:f_twice_diff}-\ref{asp:f_bound}, from [\citenum{friedlander2012hybrid}, Lemma~2.1], we obtain 
    \begin{equation}
        \mathbb{E}[F(\boldsymbol{\theta}^{(t+1)})]\!\leq\! \mathbb{E}[F(\boldsymbol{\theta}^{(t)})]\!-\!\frac{1}{2 \omega} \mathbb{E}[\|\nabla F(\boldsymbol{\theta}^{(t)})\|_2^{2}] + \!\frac{1}{2\omega} \mathbb{E}[\|\mathbf{e}^{(t)}\|_2^{2}].
        \label{eq:Lemma_upper_bound}
    \end{equation}
    Based on \eqref{eq:ErrorTwoPart}, $\mathbb{E}[\|\mathbf{e}^{(t)}\|_2^{2}]$ satisfies
    \begin{equation}
        \mathbb{E}[\|\mathbf{e}^{(t)}\|_2^{2}] = \mathbb{E}[\|\mathbf{e}_{\mathrm{sp}}^{(t)} + \mathbf{e}_{\mathrm{com}}^{(t)}\|_2^{2}]
        \leq 2(\mathbb{E}[\|\mathbf{e}_{\mathrm{sp}}^{(t)}\|_2^2] + \mathbb{E}[\|\mathbf{e}_{\mathrm{com}}^{(t)}\|_2^2]).
        \label{eq:MSE_two_part}
    \end{equation}
    By plugging \eqref{eq:MSE_two_part} into \eqref{eq:Lemma_upper_bound}, we obtain \eqref{eq:Lemma_upper_bound_two_parts}.
\end{proof}
From Lemma~\ref{Lemma:Error_up}, we see that the upper bound on $F(\cdot)$ is related to two terms, i.e., $\mathbb{E}[\|\mathbf{e}_{\mathrm{sp}}^{(t)}\|_2^2]$ and $\mathbb{E}[\|\mathbf{e}_{\mathrm{com}}^{(t)}\|_2^2]$. We next derive upper bounds on $\mathbb{E}[\|\mathbf{e}_{\mathrm{sp}}^{(t)}\|_2^2]$ and $\mathbb{E}[\|\mathbf{e}_{\mathrm{com}}^{(t)}\|_2^2]$. 
We start with $\mathbb{E}[\|\mathbf{e}_{\mathrm{sp}}^{(t)}\|_2^2]$.
\begin{lemma}
    \label{Lemma:Error_sp}
    The sparsification MSE $\mathbb{E}[\|\mathbf{e}_{\mathrm{sp}}^{(t)}\|_2^2]$ satisfies the following inequality:
    \begin{align}
        \mathbb{E}[\|\mathbf{e}_{\mathrm{sp}}^{(t)}\|_2^2] \leq 4 M\frac{\lambda^2}{(1-\lambda)^2}
        (\chi_1 + \chi_2\mathbb{E}[\|\nabla F(\boldsymbol{\theta}^{(t)})\|_2^{2}]),
        \label{eq:bound_sp_MSE}
    \end{align}
    where $\lambda$ is the sparsity ratio of the local gradient defined below \eqref{eq:g_loc_sparsity}.
\end{lemma}
\begin{proof}
    See Appendix \ref{Appendix:Prf_Lemma:Error_sp}.
\end{proof}

We next derive an upper bound of $\mathbb{E}[\|\mathbf{e}_{\mathrm{com}}^{(t)}\|_2^2]$ based on the state evolution.
The state evolution tracks the recovery error of Turbo-CS through the iteration between the variance transfer functions of the two modules\cite{ma2015performance}.
Specifically, the variance transfer functions are obtained by combining \eqref{eq:v_post_A}, \eqref{eq:v_ext_A}, \eqref{eq:v_post_B}, and \eqref{eq:v_ext_B}:
\begin{align}
    z^{(t)} & = \left( \frac{1}{\kappa} \dot (v^{(t)} + \sigma_{w}^{(t)}) - v^{(t)} \right)^{-1}, \label{eq:var_z} \\
    v^{(t)} & = \left( \frac{1}{\mmse(z^{(t)})} - z^{(t)} \right)^{-1}, \label{eq:var_v} 
\end{align}
where $z^{(t)} = \frac{1}{v_{\mathrm{B}}^{\mathrm{pri}(t)}}$, $v^{(t)} = v_{\mathrm{A}}^{\mathrm{pri}(t)}$, $\kappa$ is the compression ratio defined below \eqref{eq:g_compress}, $\sigma_{w}^{(t)}$ is defined in \eqref{eq:error_mimo} and $\mmse(z^{(t)}) = \mathbb{E}[\var[g^{\mathrm{no}(t)} | g^{\mathrm{no}(t)} + \delta^{(t)}]]$, with $g^{\mathrm{no}(t)}$ following the distribution in \eqref{eq:g_agg_sparsity}, and $\delta^{(t)} \sim \mathcal{CN}(0, 1/z^{(t)})$.
At communication round $t$, when the iteration between \eqref{eq:var_z} and \eqref{eq:var_v} converges, the variance $v^{(t)}$ achieves a fixed point $v^{\star(t)}$. 
The fixed point $v^{\star(t)}$ represents the recovery MSE of the sparse signal $\mathbf{g}^{\mathrm{no}(t)}$ in the compressed sensing problem \eqref{eq:compressed_problem} given by
\begin{equation}
    v^{\star(t)}(\mathbf{F}^{(t)}, \{\mathbf{P}_{m}^{(t)}\}) = \mathbb{E}[ \|\hat{\mathbf{g}}^{\mathrm{no}(t)} - \mathbf{g}^{\mathrm{no}(t)} \|_\mathrm{2}^2 ],
    \label{eq:fixed_point}
\end{equation}
where $v^{\star(t)}$ is a function w.r.t. $\mathbf{F}^{(t)}$ and $\{\mathbf{P}_{m}^{(t)}\}$ due to the definition of $\sigma_{w}^{(t)}$ in \eqref{eq:error_mimo}. 
The upper bound of $\mathbb{E}[\|\mathbf{e}_{\mathrm{com}}^{(t)}\|_2^2]$ is given in the following lemma.
\begin{lemma}
    \label{Lemma:Error_comm}
    The communication MSE $\mathbb{E}[\|\mathbf{e}_{\mathrm{com}}^{(t)}\|_2^2]$ satisfies the following inequality:
    \begin{align}
        \mathbb{E}[\|\mathbf{e}_{\mathrm{com}}^{(t)}\|_2^2] 
        \leq \sigma^{(t)} v^{\star(t)}(\mathbf{F}^{(t)}, \{\mathbf{P}_{m}^{(t)}\}).
        \label{eq:bound_comm_MSE}
    \end{align}
\end{lemma}
\begin{proof}
    Combining \eqref{eq:g_scaling}, \eqref{eq:ErrorTwoPart} and \eqref{eq:fixed_point} yields \eqref{eq:bound_comm_MSE}.
\end{proof}

We proceed to derive an upper bound for the expected difference between the loss function $F(\boldsymbol{\theta}^{(t+1)})$ at round $(t+1)$ and the optimal $F(\boldsymbol{\theta}^{\star})$, i.e., $\mathbb{E}[F(\boldsymbol{\theta}^{(t+1)}) - F(\boldsymbol{\theta}^{\star})]$. Applying Lemmas~\ref{Lemma:Error_up}-\ref{Lemma:Error_comm}, we obtain the following theorem.
\begin{theorem}
    \label{Theorem:Upperbound}
    Under Assumptions~\ref{asp:f_twice_diff}-\ref{asp:f_bound}, $\mathbb{E}[F(\boldsymbol{\theta}^{(t+1)}) - F(\boldsymbol{\theta}^{\star})]$ is upper bounded by
    \begin{align}
        \mathbb{E}[F(\boldsymbol{\theta}^{(t+1)}) - F(\boldsymbol{\theta}^{\star})] 
        \leq \mathbb{E}[F(\boldsymbol{\theta}^{(0)}) - F(\boldsymbol{\theta}^{\star})]
        \Psi^{t+1} + \sum\nolimits_{\tau = 1}^t \Psi^{t - \tau} \mathcal{C}^{(\tau)}(\mathbf{F}^{(t)}, \{\mathbf{P}_{m}^{(t)}\}), 
        \label{eq:F_dis_t_step}
    \end{align}
    where $\Psi$ is given by
    \begin{equation}
        \Psi = 1 - \frac{\mu}{\omega} + \frac{2 \mu \chi_2}{\omega } \cdot\frac{4 M \lambda^2}{(1-\lambda)^2},
        \label{eq:Psi}
    \end{equation}
    and $\mathcal{C}^{(t)}(\mathbf{F}^{(t)}, \{\mathbf{P}_{m}^{(t)}\})$ is given by
    \begin{equation}
        \mathcal{C}^{(t)}(\mathbf{F}^{(t)}, \{\mathbf{P}_{m}^{(t)}\}) = \frac{1}{\omega} \left( \frac{4 \chi_1 M \lambda^2}{(1-\lambda)^2} + \sigma v^{\star(t)}(\mathbf{F}^{(t)}, \{\mathbf{P}_{m}^{(t)}\}) \right).
        \label{eq:mathcal_C}
    \end{equation}
    When $\Psi < 1$, we have
    \begin{align}
        \lim_{T\rightarrow\infty}\mathbb{E}[F(\boldsymbol{\theta}^{(T)}) - F(\boldsymbol{\theta}^{\star})] 
        \leq \sum\nolimits_{\tau=1}^{T} \mathcal{C}^{(t)}(\mathbf{F}^{(t)}, \{\mathbf{P}_{m}^{(t)}\}).
        \label{eq:F_dis_limit}
    \end{align}
\end{theorem}
\begin{proof}
    See Appendix \ref{Appendix:Prf_Theorem:Upperbound}.
\end{proof}
Theorem~\ref{Theorem:Upperbound} shows that the upper bound of 
$\mathbb{E}[F(\boldsymbol{\theta}^{(t+1)})\!-\!F(\boldsymbol{\theta}^{\star})]$ 
is monotonic in $\mathcal{C}^{(t)}(\mathbf{F}^{(t)}, \{\mathbf{P}_{m}^{(t)}\})$, where $\mathcal{C}^{(t)}(\!\mathbf{F}^{(t)},\!\{\mathbf{P}_{m}^{(t)}\}\!)$ 
is a function w.r.t. the sparsity ratio $\lambda$ for sparse-coding, and the matrices $(\mathbf{F}^{(t)},\!\{\mathbf{P}_{m}^{(t)}\})$ for MIMO multiplexing. 
As seen in Theorem~\ref{Theorem:Upperbound}, reducing $\mathcal{C}^{(t)}(\mathbf{F}^{(t)},
\!\{\mathbf{P}_{m}^{(t)}\}\!)$ is desirable because it leads a faster convergence rate and a lower convergent point in \eqref{eq:F_dis_limit}. 
This motivates us to minimize $\mathcal{C}^{(t)}\!(\mathbf{F}^{(t)},\!\{\!\mathbf{P}_{m}^{(t)}\!\}\!)$ at each round $t$. 
Based on \eqref{eq:mathcal_C}, the properties of $\mathcal{C}^{(t)}\!(\mathbf{F}^{(t)},\!\{\!\mathbf{P}_{m}^{(t)}\!\}\!)$ are given in the following theorem.
\begin{theorem}
    \label{Theorem:fixed_point}
    (i) $\mathcal{C}^{(t)}(\mathbf{F}^{(t)}, \{\mathbf{P}_{m}^{(t)}\})$ is a monotonically decreasing function of the compressed gradient aggregation MSE $\sigma_{w}^{(t)}$. 
    
    (ii) $\mathcal{C}^{(t)}(\mathbf{F}^{(t)}, \{\mathbf{P}_{m}^{(t)}\})$ is a monotonically increasing function of the compression ratio $\kappa$.
    
    (iii) For a fixed number of channel uses (i.e., $K$ is fixed), the optimal number of multiplexed data streams to minimize $\mathcal{C}^{(t)}(\mathbf{F}^{(t)}, \{\mathbf{P}_{m}^{(t)}\})$ is given by $N_{\mathrm{S}} = \min\{N_{\mathrm{R}}, N_{\mathrm{T}} \}$.
\end{theorem}
\begin{proof}
    Please refer to Appendix \ref{Appendix:Prf_Theorem:fixed_point}.
\end{proof}
From Theorem~\ref{Theorem:fixed_point}, the problem of minimizing $\mathcal{C}^{(t)}(\mathbf{F}^{(t)}, \{\mathbf{P}_{m}^{(t)}\})$ is equivalent to the problem of minimizing $\sigma_{w}^{(t)}$. 
Besides, Theorem~\ref{Theorem:fixed_point} gives the optimal number of multiplexed streams $N_{\mathrm{S}}$ to minimize $\mathcal{C}^{(t)}(\mathbf{F}^{(t)}, \{\mathbf{P}_{m}^{(t)}\})$.
In the following, we design a low-complexity algorithm to solve this problem in Section~V, and verify the optimal $N_{\mathrm{S}}$ via numerical results in Section~VI. 

\section{System Optimization}
In this section, we first formulate the optimization problem to minimize the compressed gradient aggregation MSE $\sigma_{w}^{(t)}$, and then propose a low-complexity algorithm based on alternating optimization (AO) and alternating direction method of multipliers (ADMM) to solve the problem. 

\subsection{Problem Formulation}
We formulate the problem of minimizing $\sigma_{w}^{(t)}$ as (P1):
\begin{subequations}
    \begin{align}
        \text{(P1) :} \min_{\mathbf{F}^{(t)}, \{\mathbf{P}_{m}^{(t)}\}} \quad 
        & \sigma_{w}^{(t)}
        \label{eq:P1_obj} \\
        \mathrm{s.t.} \quad
        & \| \mathbf{P}_{m}^{(t)} \|_\mathrm{F}^2 \leq P_0, m \in [M]. \label{constraint_u}
    \end{align}
\end{subequations}
We note that since (P1) is solved at each round $(t)$, we omit the superscript $t$ in the description of this section for brevity.  
To further derive an explicit expression of $\sigma_{w}$, we expand \eqref{eq:error_mimo}:
\begin{align}
    \sigma_{w}
    & \overset{(a)}{=} \mathbb{E}\Big[ \big\| \sum\nolimits_{m=1}^M ( \mathbf{F} \mathbf{H}_{m} \mathbf{P}_{m} - q_{m} \mathbf{I}) \mathbf{G}_{m} \big\|_\mathrm{F}^2\Big]  + K \sigma_{\mathrm{noise}}\left\| \mathbf{F} \right\|_\mathrm{F}^2 \\
    & \overset{(b)}{=} K \sum\nolimits_{m=1}^M \sum\nolimits_{m^\prime=1}^M 
    \rho_{m^\prime, m} \tr \left ( \mathbf{P}_{m}^\mathrm{H} \mathbf{H}_{m}^\mathrm{H} \mathbf{F}^\mathrm{H} \mathbf{F} \mathbf{H}_{m^\prime} \mathbf{P}_{m^\prime} \right. \notag \\
    & \quad \quad \left.- q_{m^\prime} \mathbf{P}_{m}^\mathrm{H} \mathbf{H}_{m}^\mathrm{H} \mathbf{F}^\mathrm{H} - q_{m} \mathbf{F} \mathbf{H}_{m^\prime} \mathbf{P}_{m^\prime} + q_{m^\prime} q_{m} \mathbf{I} \right ) + K \sigma_{\mathrm{noise}}\left\| \mathbf{F} \right\|_\mathrm{F}^2, \label{eq:P1_obj_expand}
\end{align}
where step (a) is obtained by plugging the definitions of $\mathbf{R}$ and $\hat{\mathbf{R}}$ in \eqref{eq:receive_data} into \eqref{eq:error_mimo}, and in step (b) we denote the spatial correlation factor as $\rho_{m^\prime, m} = \mathbb{E}[\tr(\mathbf{G}_{m^\prime} \mathbf{G}_{m}{}^\mathrm{H})] \in [0,1], \forall m, m^\prime \in [M]$\footnote{
The spatial correlation between the local gradients of different devices is influenced by various factors, such as the data distribution, and the network topology. For more detailed discussions on spatial correlation, please refer to \cite{zhong2022over}, \cite{rodio2023federated}.
}. Based on \eqref{eq:P1_obj_expand}, (P1) is equivalently transferred to (P2):
\begin{subequations}
    \begin{align}
        \text{(P2) :} \min_{\mathbf{F}, \{\mathbf{P}_{m}\}} \quad 
        & \sum\nolimits_{m=1}^M \sum\nolimits_{m^\prime=1}^M 
        \rho_{m^\prime, m} \tr \left ( \mathbf{P}_{m}^\mathrm{H} \mathbf{H}_{m}^\mathrm{H} \mathbf{F}^\mathrm{H} \mathbf{F} \mathbf{H}_{m^\prime} \mathbf{P}_{m^\prime} \right. \notag \\
        & \qquad \qquad \left. - q_{m^\prime} \mathbf{P}_{m}^\mathrm{H} \mathbf{H}_{m}^\mathrm{H} \mathbf{F}^\mathrm{H} - q_{m} \mathbf{F} \mathbf{H}_{m^\prime} \mathbf{P}_{m^\prime} \right ) + \sigma_{\mathrm{noise}}\left\| \mathbf{F} \right\|_\mathrm{F}^2 \label{eq:P2_obj} \\
        \mathrm{s.t.} \quad
        & \| \mathbf{P}_{m} \|_\mathrm{F}^2 \leq P_0, m \in [M], \notag
    \end{align}
\end{subequations}
(P2) is an optimization problem w.r.t. the post-processing matrix $\mathbf{F}$ and the precoding matrices $\{\mathbf{P}_{m}\}$, respectively. 
However, (P2) is non-convex due to the coupling of $\mathbf{F}$ and $\{\mathbf{P}_{m}\}$. 
We adopt the AO framework to solve (P2). The AO framework contains two steps:
at Step 1, by fixing the precoding matrices $\{\mathbf{P}_{m}\}$, we can obtain a closed form of the optimal post-processing matrix $\mathbf{F}$; at Step 2, by fixing $\mathbf{F}$ and $\{\mathbf{P}_{m^\prime}\}_{m \neq m^\prime}$, we design a low-complexity algorithm to optimize $\mathbf{P}_{m}$ by using ADMM. 
The iteration of the two steps continues until the objective function in \eqref{eq:P2_obj} reaches a stable point. 
The details are discussed in what follows.

\subsection{AO-ADMM Algorithm}
We now describe the proposed AO-ADMM algorithm to solve (P2), in the following two alternating steps.
\subsubsection{Step~1: Optimizing \texorpdfstring{$\mathbf{F}$}{F} by fixing \texorpdfstring{$\{\mathbf{P}_{m}\}$}{Um} }
With fixed $\{\mathbf{P}_{m}\}$, (P2) reduces to
\begin{equation}
    \text{(P3) :} \min_{\mathbf{F}} \quad \eqref{eq:P2_obj}. \notag
\end{equation}
(P3) is a convex quadratic programming (QP) problem that has no constraints and depends on $\mathbf{F}$. It can be solved by setting $\nabla_{\mathbf{F}} \eqref{eq:P2_obj} = \mathbf{0}$,  
and the optimal $\mathbf{F}$ is given by
\begin{align}
    \mathbf{F}^\star
    & = \left ( \sum_{m=1}^M \sum_{m^\prime=1}^M q_{m^\prime} \rho_{m^\prime, m}  \mathbf{P}_{m}^\mathrm{H} \mathbf{H}_{m}^\mathrm{H} \right ) \left ( \sum_{m=1}^M \sum_{m^\prime=1}^M \rho_{m^\prime, m} 
    \mathbf{H}_{m^\prime} \mathbf{P}_{m^\prime} \mathbf{P}_{m}^\mathrm{H} \mathbf{H}_{m}^\mathrm{H} + \sigma_{\mathrm{noise}}\mathbf{I}\right )^{-1}.
    \label{eq:update_F}
\end{align}

\subsubsection{Step~2: Optimizing \texorpdfstring{$\mathbf{P}_{m}$}{Pm} by fixing \texorpdfstring{$\mathbf{F}$}{F} and \texorpdfstring{$\{\mathbf{P}_{m^\prime}\}_{m^\prime \neq m}$}{Pmprime}}
With fixed $\mathbf{F}$ and $\{\mathbf{P}_{m^\prime}\}_{m^\prime \neq m}$, we obtain the following subproblem to optimize $\mathbf{P}_{m}$:
\begin{align*}
    \text{(P4)}: 
    \min_{\mathbf{P}_{m}} 
    \quad \tr(\mathbf{P}_{m}^\mathrm{H} \mathbf{B}_{m}  \mathbf{P}_{m}) - 2 \Re \{ \tr( \mathbf{C}_{m} \mathbf{P}_{m} ) \}, 
    \quad \mathrm{s.t.} \quad 
    \tr( \mathbf{P}_{m}^\mathrm{H} \mathbf{P}_{m} ) \leq P_0,
\end{align*}
where $\mathbf{B}_{m} $ and $\mathbf{C}_{m} $ are respectively defined by
\begin{align}
    \label{eq:compute_B_C}
    \mathbf{B}_{m} 
    = \mathbf{H}_{m}^\mathrm{H} \mathbf{F}^\mathrm{H} \mathbf{F} \mathbf{H}_{m}, 
    \quad \mathbf{C}_{m} 
    = q_{m} \mathbf{F} \mathbf{H}_{m}  - \sum\nolimits_{m^\prime \neq m}^M 
    \rho_{m^\prime, m} ( \mathbf{P}_{m^\prime}^\mathrm{H} \mathbf{H}_{m^\prime}^\mathrm{H} \mathbf{F}^\mathrm{H} \mathbf{F} \mathbf{H}_{m} - 2 q_{m^\prime} \mathbf{F} \mathbf{H}_{m} ).
\end{align}
To solve (P4) efficiently, we propose a low-complexity algorithm based on consensus-ADMM\cite{huang2016consensus}.
We start with equivalently transferring (P4) into the following consensus-ADMM form:
\begin{subequations}
    \begin{align}
        \text{(P5)}: \min_{\mathbf{P}_{m}, \mathbf{Z}_{m}} \quad
        & \tr(\mathbf{P}_{m}^\mathrm{H} \mathbf{B}_{m} \mathbf{P}_{m}) - 2 \Re \{ \tr( \mathbf{C}_{m} \mathbf{P}_{m} )\}\\
        \mathrm{s.t.} \quad 
        & \tr( \mathbf{Z}_{m}^\mathrm{H} \mathbf{Z}_{m} ) \leq P_0, \label{eq:update_z_cons}\\
        & \mathbf{Z}_{m} = \mathbf{P}_{m} \label{eq:eq_cons},
    \end{align}
\end{subequations}
where $\mathbf{Z}_{m} \in \mathbb{C}^{N_{\mathrm{T}} \times N_{\mathrm{S}}}$ is the auxiliary variable. We define the augmented Lagrangian of (P5):
\begin{equation}
    \mathcal{L}_{\gamma}(\mathbf{P}_{m},\!\mathbf{Z}_{m},\!\mathbf{V}_{m})\!=\!\tr(\mathbf{P}_{m}^\mathrm{H} \mathbf{B}_{m} \mathbf{P}_{m})\!-\!2 \Re \{ \tr( \mathbf{C}_{m} \mathbf{P}_{m} )\} + \frac{\gamma}{2} \|\mathbf{Z}_{m} - \mathbf{P}_{m} + \mathbf{V}_{m} \|_\mathrm{F}^2 - \frac{\gamma}{2} \|\mathbf{V}_{m} \|_\mathrm{F}^2,
    \label{eq:augmented_Lagrangian}
\end{equation}
where $\gamma\!>\!0$ is the penalty parameter, and $\mathbf{V}_{m}\!\in\!\mathbb{C}^{N_{\mathrm{T}} \times N_{\mathrm{S}}}$ is the scaled dual variable corresponding to the constraint in \eqref{eq:eq_cons}. 
Based on \eqref{eq:augmented_Lagrangian}, the updating rules of $(\!\mathbf{P}_{m},\!\mathbf{Z}_{m},\!\mathbf{V}_{m}\!)$ are given by
\begin{subequations}
    \label{eq:updating_rules}
    \begin{align}
        \mathbf{P}_{m} \leftarrow & (\mathbf{B}_{m} + \gamma\mathbf{I})^{-1}\left( \mathbf{C}_{m} + \gamma(\mathbf{Z}_{m} + \mathbf{V}_{m}) \right); \label{eq:update_U}\\
        \mathbf{Z}_{m} \leftarrow & 
        \arg \min_{\mathbf{Z}_{m}} \|\mathbf{Z}_{m} - \mathbf{P}_{m} + \mathbf{V}_{m} \|_\mathrm{F}^2 \quad \mathrm{s.t.} \eqref{eq:update_z_cons}; \\
        \mathbf{V}_{m} \leftarrow & \mathbf{V}_{m} + \mathbf{Z}_{m} - \mathbf{P}_{m}.
        \label{eq:update_V}
    \end{align}
\end{subequations}
From \eqref{eq:updating_rules}, we see that both $\mathbf{P}_{m}$ and $\mathbf{V}_{m}$ have closed form updating expressions. The remaining issue is to solve the subproblem of updating $\mathbf{Z}_{m}$:
\begin{align*}
    \text{(P6)}: \min_{\mathbf{Z}_{m}} & \quad \|\mathbf{Z}_{m} - \mathbf{P}_{m} + \mathbf{V}_{m} \|_\mathrm{F}^2, 
    \quad \mathrm{s.t.} \quad \eqref{eq:update_z_cons}.
\end{align*}
The Lagrangian function of (P6) can be formulated as
\begin{equation}
    \mathcal{L}(\mathbf{Z}_{m}) = \|\mathbf{Z}_{m} - \mathbf{P}_{m} + \mathbf{V}_{m} \|_\mathrm{F}^2 + \zeta_{m} (\tr( \mathbf{Z}_{m}^\mathrm{H} \mathbf{Z}_{m} ) - P_0 ),
\end{equation}
where $\zeta_{m}$ is the dual variable corresponding to the constraint in \eqref{eq:update_z_cons}. By setting $\nabla \mathcal{L}(\mathbf{Z}_{m}) = 0$, we obtain the optimal $\mathbf{Z}_{m}$, given by
\begin{equation}
    \mathbf{Z}_{m} = \frac{1}{1 + \zeta_{m}}(\mathbf{P}_{m}-\mathbf{V}_{m}). 
    \label{eq:update_Z}
\end{equation}
By plugging \eqref{eq:update_Z} into \eqref{eq:update_z_cons}, the update of the dual variable $\zeta_{m}$ can be obtained based on the complementary slackness condition:
\begin{equation}
    \zeta_{m} = \max \Big \{\|\mathbf{P}_{m}-\mathbf{V}_{m}\|_\mathrm{F} / \sqrt{P_0} - 1, 0 \Big \}.
    \label{eq:update_zeta}
\end{equation}

To sum up, the proposed algorithm to optimize $\mathbf{P}_{m}$ is to alternatively update $\mathbf{P}_{m}$ via \eqref{eq:update_U}, the auxiliary variable $\mathbf{Z}_{m}$ via \eqref{eq:update_Z}, and the dual variables $\mathbf{V}_{m}$ via \eqref{eq:update_V} and $\zeta_{m}$ via \eqref{eq:update_zeta}, until the gap $\|\mathbf{P}_{m} - \mathbf{Z}_{m}\|_\mathrm{F}$ less than a predetermined convergence threshold $\epsilon$. 
We summarize the proposed AO-ADMM algorithm to optimize $(\mathbf{F}, \{\mathbf{P}_{m}\})$ in Algorithm~\ref{alg:AO-ADMM}.

\begin{algorithm}[htb]
\caption{ Proposed AO-ADMM Algorithm to optimize $(\mathbf{F}, \{\mathbf{P}_{m}\})$} 
\label{alg:AO-ADMM} 
\begin{algorithmic}[1] 
\REQUIRE $\{q_{m}, \mathbf{H}_{m}, \boldsymbol{\rho} | m \in [M]\}$, $I_{\max}$ and $\epsilon$. 
\STATE {\textbf{Initialization:}} Let both $\mathbf{F}$ and $\{\mathbf{P}_{m}\}$ feasible for (P2);
\FOR{$\tau \in [I_{\max}]$}
    \STATE {\textbf{Step~1:}}
    \STATE Update $\mathbf{F}$ via \eqref{eq:update_F};
    \STATE {\textbf{Step~2:}} 
    \FOR{$m \in [M]$ }
        \STATE Compute $\mathbf{B}_{m}$, and $\mathbf{C}_{m}$ based on \eqref{eq:compute_B_C};
        \REPEAT
            \STATE Update $\mathbf{P}_{m}$ via \eqref{eq:update_U};
            \STATE Update $\mathbf{Z}_{m}$ via \eqref{eq:update_Z} and update $\zeta_{m}$ via \eqref{eq:update_zeta};
            \STATE Update $\mathbf{V}_{m}$ via \eqref{eq:update_V};
        \UNTIL{$\|\mathbf{P}_{m} - \mathbf{Z}_{m}\|_\mathrm{F} \leq \epsilon$; }
    \ENDFOR
\ENDFOR
\ENSURE $(\mathbf{F}, \{\mathbf{P}_{m}\})$.
\end{algorithmic}
\end{algorithm}

\subsection{Complexity and Convergence Analysis}
We first briefly discuss the computational complexity involved in Algorithm~\ref{alg:AO-ADMM}. We assume that Algorithm~\ref{alg:AO-ADMM} needs $I_{\max}$ iterations to converge. 
In each iteration $\tau$, the algorithm optimizes $\mathbf{F}$ and $\{\mathbf{P}_{m}\}$ alternatively. 
On one hand, for each $\mathbf{P}_{m}$, the algorithm solves (P5) by alternatively updating $\mathbf{P}_{m}$, $\mathbf{Z}_{m}$, $\mathbf{V}_{m}$, and $\zeta_{m}$ until the gap $\|\mathbf{P}_{m} - \mathbf{Z}_{m}\|_\mathrm{F}$ converges, which has a complexity of $\mathcal{O}(N^3\log(1/\epsilon))$ with $N = \max\{N_{\mathrm{T}}, N_{\mathrm{R}}, N_{\mathrm{S}}\}$.
On the other hand, the algorithm computes the optimal $\mathbf{F}$ via \eqref{eq:update_F}, resulting a worst complexity of $\mathcal{O}(N^3)$.
Therefore, the overall complexity of Algorithm~\ref{alg:AO-ADMM} is given by $\mathcal{O}(I_{\max}(M N^3\log(1/\epsilon) + N^3)$. 
In contrast, the complexity of solving (P5) with the interior point method \cite{wright1997primal} is $\mathcal{O}((N_{\mathrm{T}} N_{\mathrm{S}})^{3.5} \log(1/\epsilon)) \approx \mathcal{O}(N^{7} \log(1/\epsilon))$, which is much higher than that of Algorithm~\ref{alg:AO-ADMM}.

The convergence of Algorithm~\ref{alg:AO-ADMM} is guaranteed by noting that the monotonically non-decreasing objective of (P2) in the AO iterations, and the convergence of solving subproblem (P5) with ADMM in each AO iteration\cite{huang2016consensus}.

\section{Numerical Experiments}

\subsection{Simulation Setup}

We simulate a scenario in three dimensions (3-D) as depicted in Fig.~\ref{fig:BS}. 
The cylindrical coordinates $(\delta, \phi, \xi)$ represent the point locations, where $\delta$ is the radial distance from the point to the origin, $\phi$ is the azimuth angle, and $\xi$ is the height.  
The PS is located at $(0, 0, 10)$. 
The $m$-th device is located at $(\delta_{m}, \phi_{m}, 0)$, where $\delta^2_{m}$ is uniformly distributed over $[0, \Delta^2]$, and $\phi_{m}$ is uniformly distributed over $[0, 2\pi)$. 
The channel model is given by $\mathbf{H}_{m} = \sqrt{G_{\mathrm{R}} G_{\mathrm{T}} \ell \tilde{\delta}_{m}^{-\beta}} \mathbf{\tilde{H}}_{m}$ \cite{goldsmith2005wireless}, 
where each entry of $\mathbf{\tilde{H}}_{m}$ is i.i.d. drawn from $\mathcal{CN}(0, 1)$, $G_{\mathrm{T}}$ is the antenna gain at each device, $G_{\mathrm{R}}$ is the antenna gain at the PS, $\ell$ denotes the path loss at the reference distance $\SI{1}{m}$ \cite{wu2019intelligent}, $\tilde{\delta}_{m} = \sqrt{\delta_{m}^2 + 10^2}$ denotes the distance between the $m$-th device and the PS, and $\beta$ denotes the path loss exponent. We list other simulation settings in Table \ref{SimuPara}.

\begin{figure}[htbp]
    \centering
    \includegraphics[width=0.5\linewidth]{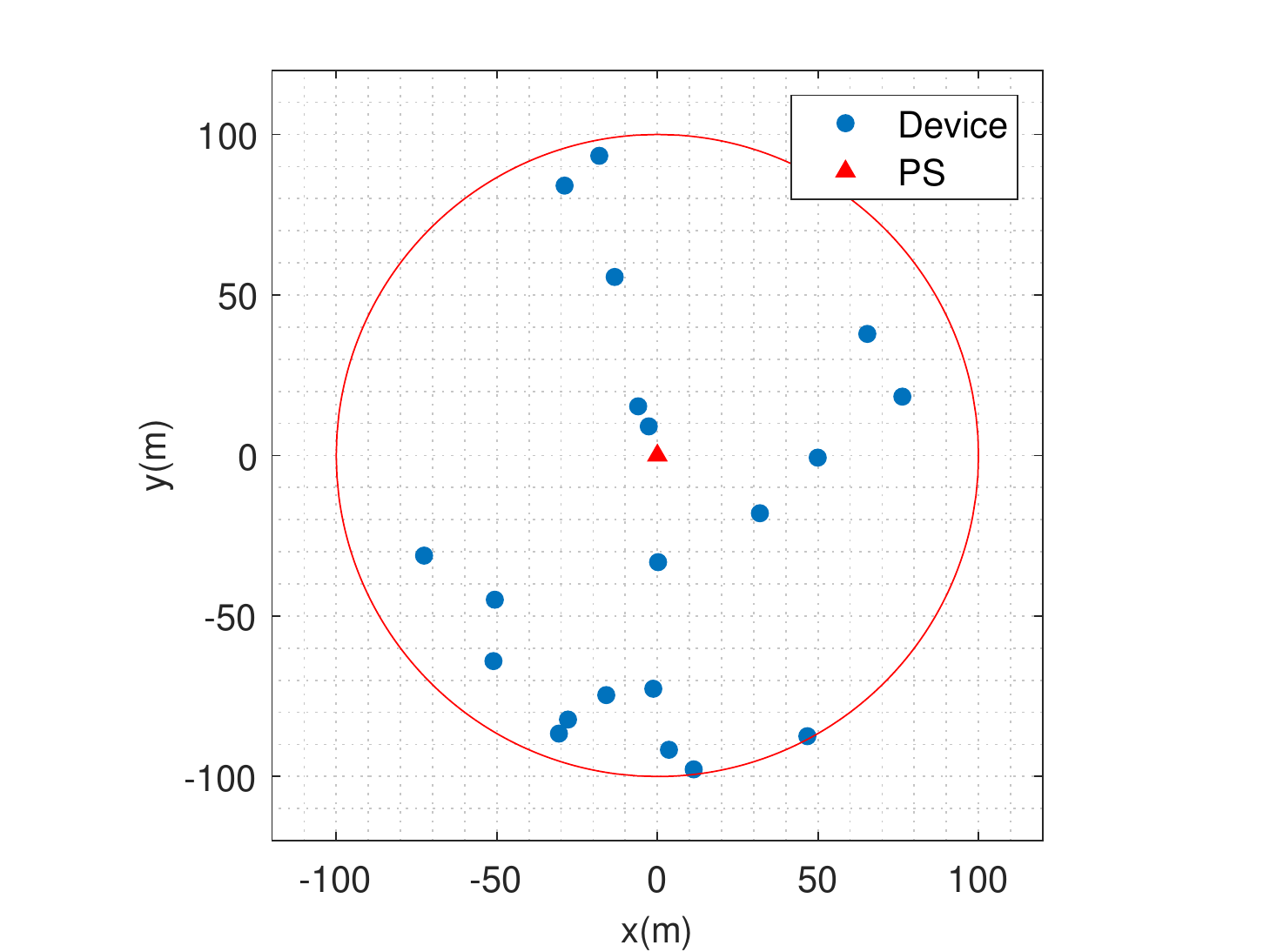}
    \caption{ A depiction of the positions of the PS and the devices in the OA-FL. }
    \label{fig:BS}
\end{figure}
\begin{table}[ht]
\caption{Simulation Parameters and Values}
\small 
\centering
\begin{tabular}{|c|c||c|c||c|c||c|c|}
\hline    
Parameter&Value &Parameter&Value &Parameter&Value &Parameter&Value\\
\hhline{|--||--||--||--|}    
$P_0$&$\SI{0.1}{W}$       &$\sigma_{\mathrm{noise}}$&$\SI{-90}{dBm}$        &$I_{\max}$&$50$      &$\Delta$&$\SI{100}{m}$     \\ 
\hline    
$\beta$&$3.8$         &$\ell$&$\SI{-60}{dB}$  &$G_{\mathrm{T}}$&$\SI{5}{dBi}$ &$G_{\mathrm{R}}$&$\SI{5}{dBi}$\\ 
\hline    
$Q$&$60,000$  &$Q_m$&$3,000$ &$\gamma$ &$\SI{1e-4}{}$ &$\epsilon$&$\SI{1e-4}{}$ \\ 
\hline 
\end{tabular}
\label{SimuPara}
\end{table}


We use a neural network that has a $5\times5$ convolution layer with $16$ feature maps, a $5\times5$ convolution layer with $32$ feature maps, and a $50$-unit fully connected layer with ReLu activation function. 
The final layer is a softmax classifier. Each convolution layer is followed by a $2\times2$ max pooling layer, a ReLu activation layer and a batch normalization layer.
The cross-entropy loss is used as the loss function. The total number of trainable parameters is $D=39604$. 
Each round of local updates consists of $10$ mini-batches with a batch size of $300$. We train our network on two datasets: MNIST\cite{lecun2010mnist} and FMNIST\cite{xiao2017fashion}. 
We generate the local datasets $\{\mathcal{A}_m\}$ on the devices in the following two ways: \textbf{i.i.d.}, where we shuffle all the data samples and assign them evenly to all the devices; \textbf{non-i.i.d.}, where we pick $4$ classes and draw $\frac{Q}{4 M}$ samples randomly from each chosen class for each device.

\subsection{Performance Comparisons}
We first present the convergence performance of Algorithm~\ref{alg:AO-ADMM} with a varying channel noise power. As shown in Fig.~\ref{fig:MSE}, the compressed gradient aggregation MSE $\sigma_{w}^{(t)}$ decreases monotonically in the iterations, and converges to a stable point within a small number of iterations on average. 
Hence, Algorithm~\ref{alg:AO-ADMM} has excellent convergence performance and fast convergence rate.
\begin{figure}[htbp]
    \centering
    \includegraphics[width = 0.75\linewidth]{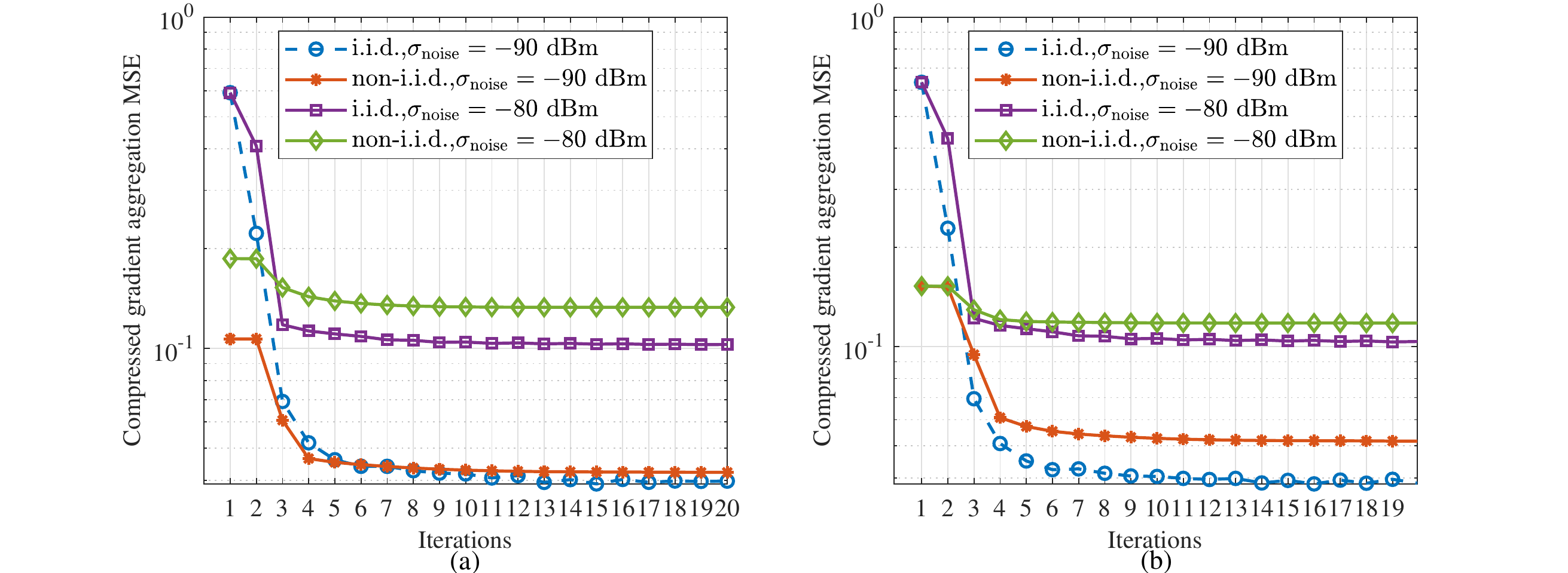}
    \caption{Compressed gradient aggregation MSE versus the number of iterations. (a): MNIST; (b): Fashion-MNIST. $M = 20$, $N_{\mathrm{S}} = 4$, $N_{\mathrm{T}} = 8$, $N_{\mathrm{R}} = 16$, $\lambda = 0.05$, and $\kappa = 0.5$. The result is averaged over $100$ communication rounds.}
    \label{fig:MSE}
\end{figure}

Then, we study the impact of the number of multiplexing streams $N_{\mathrm{S}}$ on the learning performance of the SCoM approach, where the number of transmit antennas is fixed at $N_{\mathrm{T}} = 4$, the number of receive antennas at $N_{\mathrm{R}} = 8$, the sparsity ratio at $\lambda = 0.05$, and the number of channel uses at $K = 1584$. 
The results are obtained by taking the mean of $10$ Monte Carlo simulations.
Fig.~\ref{fig:N_S} demonstrates test accuracies versus the number of multiplexing streams $N_{\mathrm{S}}$ in the case of both i.i.d. data and non-i.i.d. data on the two datasets. 
As shown in Fig.~\ref{fig:N_S}, it is seen that in the four cases, as $N_{\mathrm{S}}$ increases, the test accuracies sharply increase when $N_{\mathrm{S}} \leq N_{\mathrm{T}}$, and gradually decline when $N_{\mathrm{S}} > N_{\mathrm{T}}$. The peak of the test accuracies appears at $N_{\mathrm{S}} = N_{\mathrm{T}}$.
This is because that for a fixed number of channel uses $K$, a smaller $N_{\mathrm{S}}$ requires a smaller compression ratio $\kappa$, which results in more information loss in the sparse-coding; while a larger $N_{\mathrm{S}}$ causes a larger compressed gradient aggregation MSE $\sigma_{w}^{(t)}$ due to the bottleneck of the number of antennas.
The results are consistent with Theorem~\ref{Theorem:fixed_point} in Section-IV.
\begin{figure}[htbp]
    \centering
    \includegraphics[width = 0.8\linewidth]{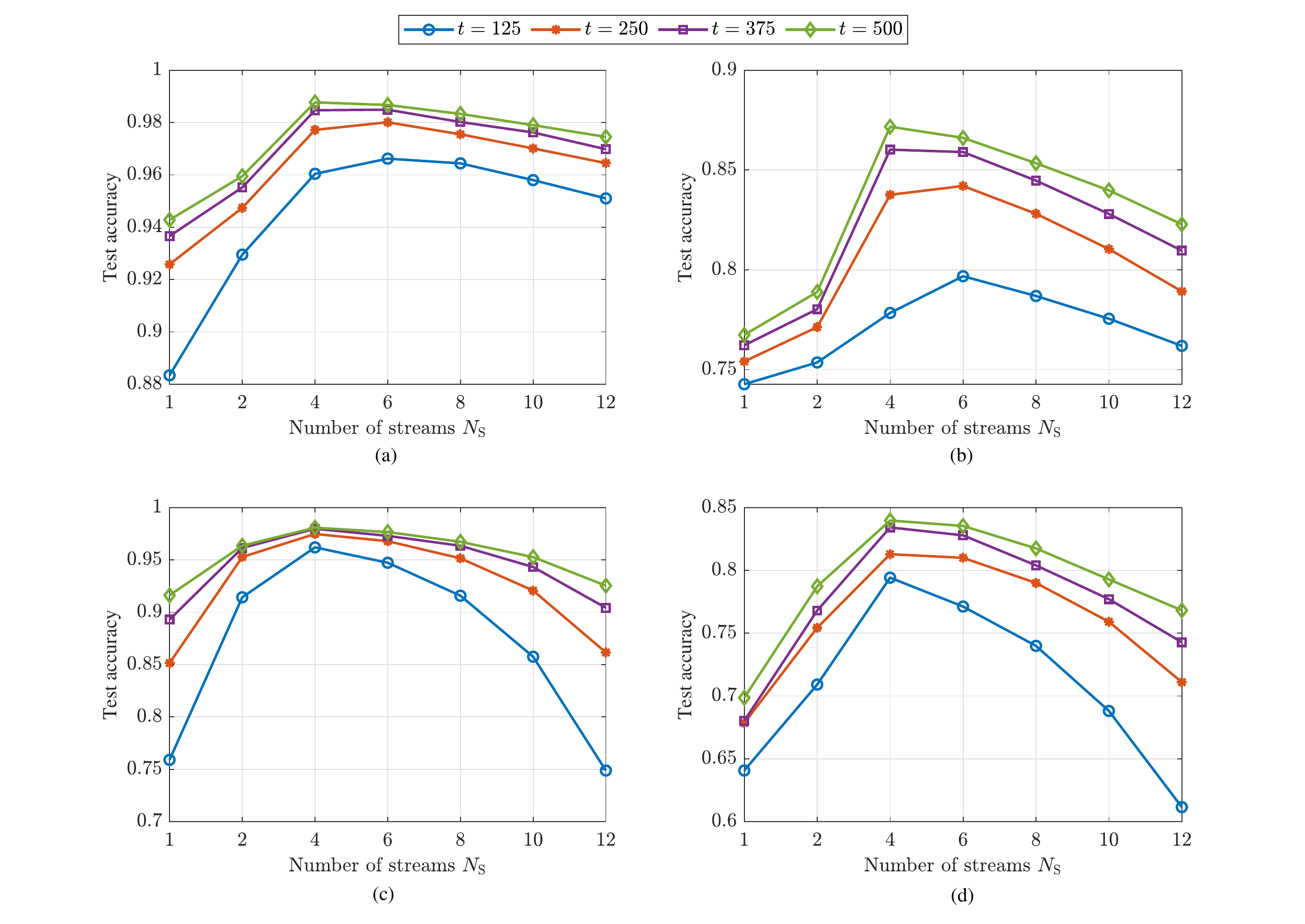}
    \caption{Test accuracy versus the number of streams on different rounds. MNIST: (a),(c); Fashion-MNIST: (b), (d); i.i.d. data: (a), (b); non-i.i.d. data: (c), (d). $M = 20$, $N_{\mathrm{T}} = 4$, $N_{\mathrm{R}} = 8$, and $K = 1584$.}
    \label{fig:N_S}
\end{figure}

We next investigate the impact of the number of receive antennas $N_{\mathrm{R}}$ at the PS on the learning performance of the SCoM approach, where the number of data streams is fixed at $N_{\mathrm{S}} = 4$ and the number of transmit antennas at $N_{\mathrm{T}} = 8$. The compression ratio $\kappa$ is set to $0.5$, the sparsity ratio $\lambda$ is set to $0.05$, and the number of channel uses is $K = 2476$. 
In Fig~\ref{fig:N_R}, it is observed that when $N_{\mathrm{R}} \leq N_{\mathrm{S}}$, a larger $N_{\mathrm{R}}$ implies a higher learning accuracy; but when $N_{\mathrm{R}} > N_{\mathrm{S}}$, the improvement on the learning accuracy by increasing $N_{\mathrm{R}}$ is neglectable. 
This is due to the fact that the case of $N_{\mathrm{R}} < N_{\mathrm{S}}$ leads to $\rank(\mathbf{H}_{m}^{(t)}) < N_{\mathrm{S}}$, which causes a large compressed gradient aggregation error; while in the case of $N_{\mathrm{R}} > N_{\mathrm{S}}$, the number of antennas is no longer the bottleneck of the system.

\begin{figure}[htbp]
    \centering
    \includegraphics[width = 0.75\linewidth]{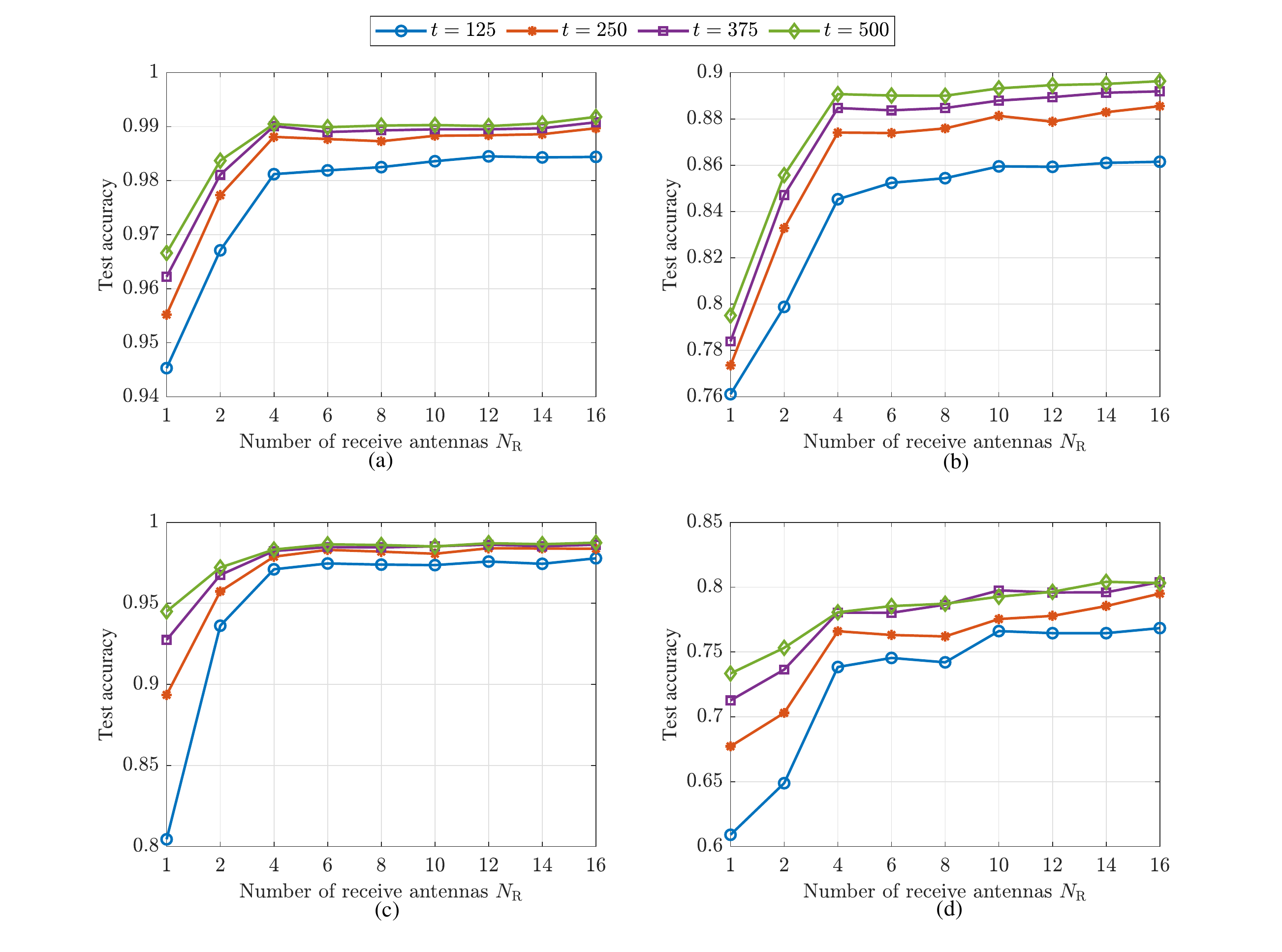}
    \caption{Test accuracy versus the number of receive antennas on different rounds. MNIST: (a),(c); Fashion-MNIST: (b), (d); i.i.d. data: (a), (b); non-i.i.d. data: (c), (d). $M = 20$, $N_{\mathrm{S}} = 4$, $N_{\mathrm{T}} = 8$, and $K = 2476$.}
    \label{fig:N_R}
\end{figure}

We now show the effect of the numbers of transmit antennas $N_{\mathrm{T}}$ at the devices on the learning performance of the SCoM approach, where the number of data streams is fixed at $N_{\mathrm{S}} = 4$ and the number of receive antennas at $N_{\mathrm{R}} = 8$. The compression ratio $\kappa$ is set to $0.5$, the sparsity ratio $\lambda$ is set to $0.05$, and the number of channel uses is $K = 2476$. 
Fig~\ref{fig:N_T} illustrates that when $N_{\mathrm{T}} < N_{\mathrm{S}}$, the test accuracy sharply increases with $N_{\mathrm{T}}$ increasing in all the four cases, since the increasing of $N_{\mathrm{T}}$ reduces the compressed gradient aggregation error. 
This is because that the case of $N_{\mathrm{T}} < N_{\mathrm{S}}$ leads to $\rank(\mathbf{H}_{m}^{(t)}) < N_{\mathrm{S}}$, which causes a large compressed gradient aggregation error.

\begin{figure}[t!]
    \centering
    \includegraphics[width = 0.75\linewidth]{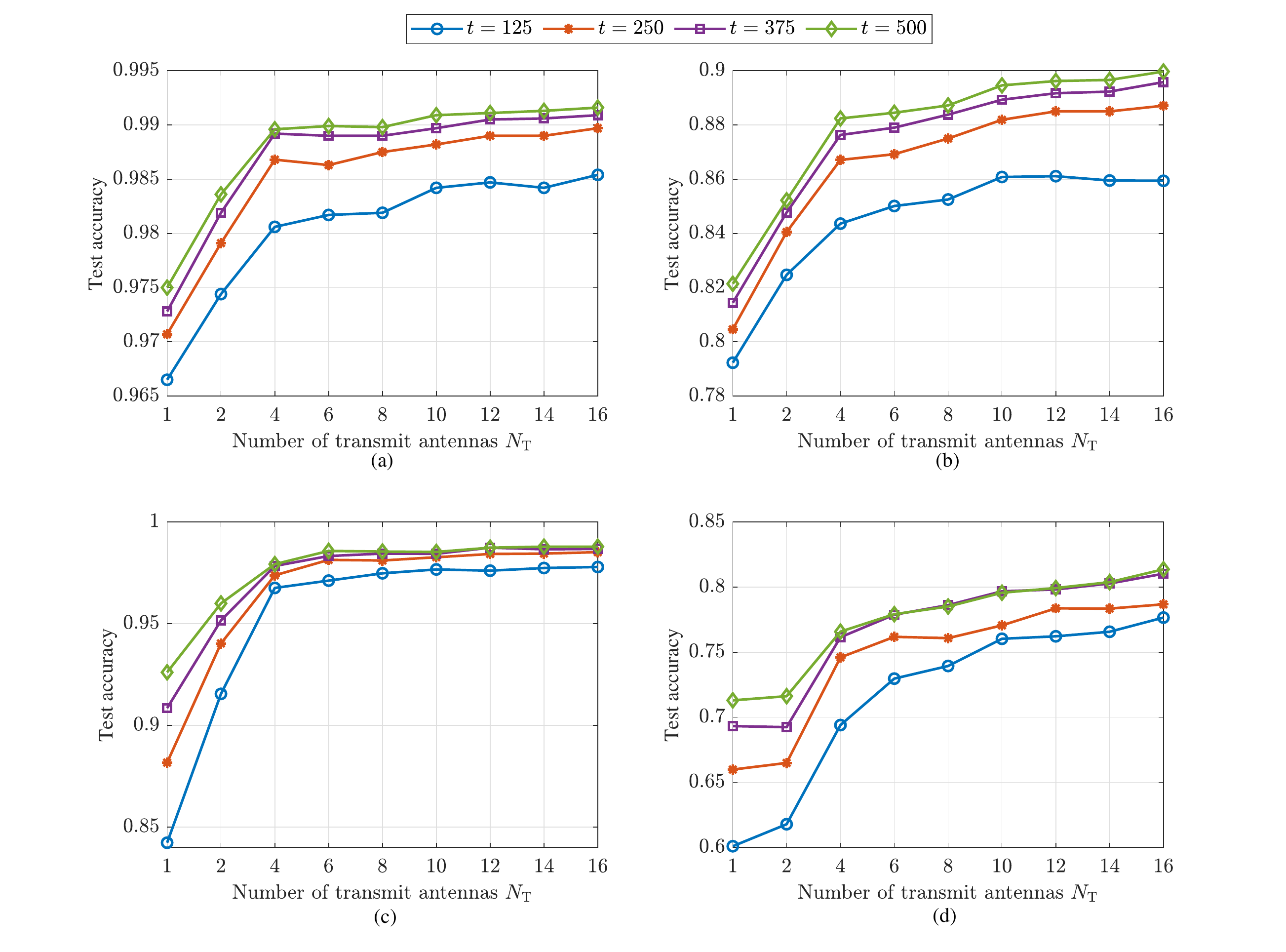}
    \caption{Test accuracy versus the number of transmit antennas on different rounds. MNIST: (a),(c); Fashion-MNIST: (b), (d); i.i.d. data: (a), (b); non-i.i.d. data: (c), (d). $M = 20$, $N_{\mathrm{S}} = 4$, $N_{\mathrm{R}} = 8$, and $K = 2476$.}
    \label{fig:N_T}
\end{figure}

To evaluate the efficiency of the SCoM method, we compare the total number of channel uses needed to train a FL task with other existing communication approaches.
There are four approaches, namely the proposed SCoM approach, zero-forcing \cite{zhu2018mimo}, antennas selection \cite{chen2018over}, and AO-based beamforming \cite{zhong2022over}. 
The first three approaches all use the MIMO multiplexing technique to transmit multiple streams in parallel, while the last one transmits a single data stream.
The numbers of antennas are set to $N_{\mathrm{T}} = 8$, $N_{\mathrm{R}} = 16$, and the number of data streams is set to $N_{\mathrm{S}} = 8$.
For the SCoM approach, the compression ratio $\kappa$ is set to $0.5$, and the sparsity ratio $\lambda$ is set to $0.05$.
Based on the above parameters, we can obtain the number of channel uses of four approaches in each round, given by $1238$, $2476$, $2476$, and $19802$, respectively.
Moreover, we denote $K(\tilde{\psi})$ as the number of channel uses needed to train the FL task until the task achieves its \textit{relative target accuracy} $\tilde{\psi}$, where $\tilde{\psi} = \frac{\psi}{\psi^{\max}} \in [0,1]$, $\psi$ is the test accuray and $\psi^{\max}$ is the maximum test accuracy.
We set $\psi_{1}^{\max} = 0.987$ for i.i.d. MNIST, $\psi_{2}^{\max} = 0.87$ for i.i.d. FMNIST, $\psi_{3}^{\max} = 0.985$ for non-i.i.d. MNIST, and $\psi_{4}^{\max} = 0.83$ for non-i.i.d. FMNIST. 

Fig.~\ref{fig:K} depicts the total number of channel uses needed by various communication approaches versus the relative target accuracy. We observe that both zero-forcing and antennas selection perform worse than the other two approaches and fail to reach $\psi^{\max}$, since they suffer incur a large compressed gradient aggregation error due to channel inversion. 
Meanwhile, compared with other baselines, the SCoM approach consumes the least number of channel uses to achieve the same learning accuracy in all cases. Compared with the AO-based beamforming, the SCoM approach only consumes $30\%\!\sim\!40\%$ of the number of channel uses. The results clearly demonstrate the outstanding performance of our proposed scheme in reducing communication overhead.

\begin{figure}[htbp]
    \centering
    \includegraphics[width = 0.75\linewidth]{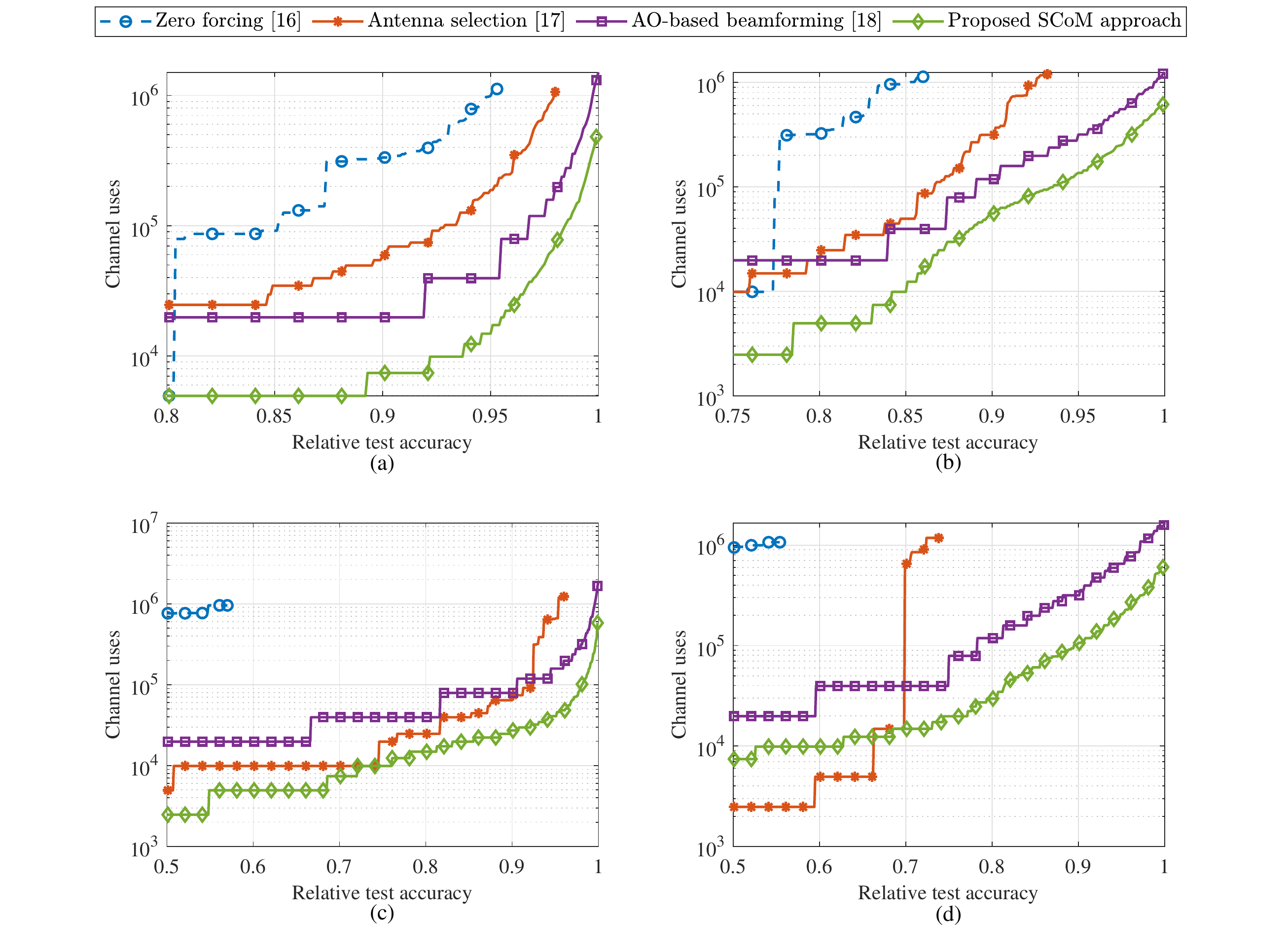}
    \caption{The number of channel uses versus relative test accuracy via different approaches. MNIST: (a),(c); Fashion-MNIST: (b), (d); i.i.d. data: (a), (b); non-i.i.d. data: (c), (d). $M = 20$, $N_{\mathrm{T}} = 8$, $N_{\mathrm{R}} = 16$, $N_{\mathrm{S}} = 8$, and $\kappa = 0.5$, $\sigma_{\mathrm{noise}} = \SI{-80}{dBm}$ for i.i.d. data, and $\sigma_{\mathrm{noise}} = \SI{-90}{dBm}$ for non-i.i.d. data.}
    \label{fig:K}
\end{figure}

\section{Conclusions}


In this paper, we proposed a novel SCoM approach based on the sparse-coding and the MIMO multiplexing to reduce the communication overhead of the gradient uploading in MIMO OA-FL. 
We derived an upper bound on the learning performance loss of the SCoM-based MIMO OA-FL scheme by characterizing the gradient aggregation error. 
Based on the analytical results, the optimal number of multiplexed data streams to minimize the upper bound on the FL learning performance loss is given by the minimum number of transmit and receive antennas. 
To minimize the gradient aggregation error, we formulated the optimization problem by optimizing the precoding matrices and the post-processing matrix. 
We proposed a low-complexity algorithm to solve the problem based on the AO framework and the ADMM algorithm.
Finally, by comparing with the state-of-the-art methods, the numerical experiments demonstrated the outstanding performance of the SCoM approach in balancing the communication overhead and the learning performance.

\appendices

\section{Proof of Lemma~\ref{Lemma:Error_sp}}
\label{Appendix:Prf_Lemma:Error_sp}
To start with, we bound $\mathbb{E}[\|\mathbf{e}_{\mathrm{sp}}^{(t)}\|_2^2]$ as 
\begin{align}
    \mathbb{E}[\|\mathbf{e}_{\mathrm{sp}}^{(t)}\|_2^2] 
    & = \mathbb{E}[\|\mathbf{g}^{(t)} - \tilde{\mathbf{g}}^{\mathrm{sp}(t)} \|_2^2] 
    \overset{\text{(a)}}{=} \mathbb{E}\bigg[ \big\|\sum\nolimits_{m=1}^M q_{m} (\mathbf{g}_{m}^{(t)} - \tilde{\mathbf{g}}_{m}^{\mathrm{sp}(t)} ) \big\|_2^2\bigg] \allowdisplaybreaks \\
    & \overset{\text{(b)}}{=} \mathbb{E}\bigg[\big\|\sum\nolimits_{m=1}^M q_{m} (\mathbf{g}_{m}^{\mathrm{ac}(t)} + \boldsymbol{\Delta}_{m}^{(t)} - \mathbf{g}_{m}^{\mathrm{ac}(t)} + \boldsymbol{\Delta}_{m}^{(t+1)}) \big\|_2^2\bigg] \allowdisplaybreaks \\
    & = \mathbb{E}\bigg[\big\|\sum\nolimits_{m=1}^M q_{m} (\boldsymbol{\Delta}_{m}^{(t)} + \boldsymbol{\Delta}_{m}^{(t+1)}) \big\|_2^2\bigg] 
    \allowdisplaybreaks \\
    & \leq 2 M \sum\nolimits_{m=1}^M q_{m} ( \mathbb{E}[\| \boldsymbol{\Delta}_{m}^{(t)} \|_2^2] + \mathbb{E}[\| \boldsymbol{\Delta}_{m}^{(t+1)} \|_2^2] )
    \label{eq:e_sp_upbound_delta}
\end{align}
where step (a) is due to $\mathbf{g}^{(t)}$ defined in \eqref{eq:IdealGradient} and $\tilde{\mathbf{g}}_{m}^{\mathrm{sp}(t)}$ defined in \eqref{eq:IdealSpGradient}, and step (b) is obtained with \eqref{eq:Accumulation} and \eqref{eq:Delta_update}.
Meanwhile, from [\citenum{amiri2020machine}, Appendix A], we have 
$ \| \boldsymbol{\Delta}_{m}^{(t)} \|_2 \leq \sum\nolimits_{\tau=0}^{t-1} \lambda^{t-\tau} \|\mathbf{g}_m^{(\tau)} \|_2 $.
Using \eqref{eq:f_bound} in Assumption~\ref{asp:f_bound}, we obtain
$\mathbb{E}[\| \boldsymbol{\Delta}_{m}^{(t)} \|_2^2 ] \leq (\chi_1 + \chi_2 \mathbb{E}[\|\nabla F(\boldsymbol{\theta}^{(t)})\|_2^{2}]) (\sum\nolimits_{\tau=0}^{t-1} \lambda^{t-\tau} )^2 $.
Using the above inequality in \eqref{eq:e_sp_upbound_delta} yields
\begin{align}
    \mathbb{E}[\|\mathbf{e}_{\mathrm{sp}}^{(t)}\|_2^2] 
    & \leq 2 M (\chi_1 + \chi_2 \mathbb{E}[\|\nabla F(\boldsymbol{\theta}^{(t)})\|_2^{2}]) \sum_{m=1}^M q_{m} \Bigg( \left( \frac{\lambda(1 - \lambda^{t+1})}{1-\lambda}\right)^2 + \left( \frac{\lambda(1 - \lambda^{t})}{1-\lambda} \right)^2 \Bigg) \notag \\
    & \leq 4 M \frac{\lambda^2}{(1-\lambda)^2} (\chi_1 + \chi_2 \mathbb{E}[\|\nabla F(\boldsymbol{\theta}^{(t)})\|_2^{2}]).
\end{align}

\section{Proof of Theorem~\ref{Theorem:Upperbound}}
\label{Appendix:Prf_Theorem:Upperbound}
By plugging \eqref{eq:bound_sp_MSE} and \eqref{eq:bound_comm_MSE} into \eqref{eq:Lemma_upper_bound}, we obtain
\begin{align}
    \mathbb{E}[F(\boldsymbol{\theta}^{(t+1)})] 
    & \leq \mathbb{E}[F(\boldsymbol{\theta}^{(t)})] - \mathbb{E}[\|\nabla F(\boldsymbol{\theta}^{(t)})\|_2^{2}] \cdot \frac{1}{\omega} \left( \frac{1}{2} - \frac{4 \chi_2 M \lambda^2}{(1-\lambda)^2}\right) + \mathcal{C}^{(t)}(\mathbf{F}^{(t)}, \{\mathbf{P}_{m}^{(t)}\}), 
    \label{eq:bound_asp_4}
\end{align}
where $\mathcal{C}^{(t)}(\mathbf{F}^{(t)}, \{\mathbf{P}_{m}^{(t)}\})$ is defined in \eqref{eq:mathcal_C}. Meanwhile, from [\citenum{friedlander2012hybrid}, eq.~(2.4)], we have 
\begin{equation}
    \mathbb{E}[\|\nabla F(\boldsymbol{\theta}^{(t)})\|_2^{2}] \geq 2 \mu \mathbb{E}[F(\boldsymbol{\theta}^{(t)}) - F(\boldsymbol{\theta}^\star)].
    \label{eq:g_lowerbound}
\end{equation}
By subtracting $F(\boldsymbol{\theta}^\star)$ on both side of \eqref{eq:bound_asp_4}, and substituting $\mathbb{E}[\|\nabla F(\boldsymbol{\theta}^{(t)})\|_2^{2}]$ with \eqref{eq:g_lowerbound}, we obtain
\begin{align}
    \mathbb{E}[F(\boldsymbol{\theta}^{(t+1)}) - F(\boldsymbol{\theta}^{\star})] 
    \leq \mathbb{E}[F(\boldsymbol{\theta}^{(t)}) - F(\boldsymbol{\theta}^{\star})]
    \Psi + \mathcal{C}^{(t)}(\mathbf{F}^{(t)}, \{\mathbf{P}_{m}^{(t)}\}), 
    \label{eq:F_dis_one_step}
\end{align}
where $\Psi$ is defined in \eqref{eq:Psi}. 
Recursively applying \eqref{eq:F_dis_one_step} for $(t+1)$ times, we obtain \eqref{eq:F_dis_t_step}.

The remaining issue is to prove \eqref{eq:F_dis_limit}. When $\Psi < 1$, we have $\Psi^{t - \tau} \mathcal{C}^{(\tau)}(\mathbf{F}^{(t)}, \{\mathbf{P}_{m}^{(t)}\}) \leq \mathcal{C}^{(t)}(\mathbf{F}^{(t)}, \{\mathbf{P}_{m}^{(t)}\})$, and $\lim_{T\rightarrow\infty}\Psi^T = 0$. Plugging the results into \eqref{eq:F_dis_t_step} yields \eqref{eq:F_dis_limit}. 

\section{Proof of Theorem~\ref{Theorem:fixed_point}}
\label{Appendix:Prf_Theorem:fixed_point}
We first prove Theorem~\ref{Theorem:fixed_point}-(i). From \eqref{eq:var_z}, we see that $z^{(t)}$ is a monotonically decreasing function of the compressed gradient aggregation MSE $\sigma_{w}^{(t)}$. 
As seen in \eqref{eq:var_v}, the increase of $z^{(t)}$ leads to the decrease of both the noise variance $1/z^{(t)}$ and $\mmse(z^{(t)})$, resulting in the increase of $v^{(t)}$.
Thus, the fixed point $v^{\star(t)}(\mathbf{F}^{(t)}, \{\mathbf{P}_{m}^{(t)}\})$ is a monotonically decreasing function of $\sigma_{w}^{(t)}$.
From \eqref{eq:mathcal_C}, $\mathcal{C}^{(t)}(\mathbf{F}^{(t)}, \{\mathbf{P}_{m}^{(t)}\})$ is also a monotonically decreasing function of $\sigma_{w}^{(t)}$.
The proof of Theorem~\ref{Theorem:fixed_point}-(ii) follows a similar logic to the proof of Theorem~\ref{Theorem:fixed_point}-(i), and we omit the details for space limitations.

We now prove Theorem~\ref{Theorem:fixed_point}-(iii). 
On one hand, when $N_{\mathrm{S}} > \min\{N_{\mathrm{R}}, N_{\mathrm{T}} \}$, $\sigma_{w}^{(t)}$ in \eqref{eq:error_mimo} significantly increases due to $\rank(\mathbf{H}_{m}^{(t)}) < N_{\mathrm{S}}$, which leads to larger values of $v^{\star(t)}(\mathbf{F}^{(t)}, \{\mathbf{P}_{m}^{(t)}\})$ and $\mathcal{C}^{(t)}(\mathbf{F}^{(t)}, \{\mathbf{P}_{m}^{(t)}\})$. 
On the other hand, when $N_{\mathrm{S}} < \min\{N_{\mathrm{R}}, N_{\mathrm{T}} \}$, a smaller $N_{\mathrm{S}}$ results in a smaller compression ratio $\kappa$ for a fixed number of channel uses $K$, which also leads to larger values of $v^{\star(t)}(\mathbf{F}^{(t)}, \{\mathbf{P}_{m}^{(t)}\})$ and $\mathcal{C}^{(t)}(\mathbf{F}^{(t)}, \{\mathbf{P}_{m}^{(t)}\})$.
Thus, the optimal $N_{\mathrm{S}}$ is given by $N_{\mathrm{S}} = \min\{N_{\mathrm{R}}, N_{\mathrm{T}} \}$. 

\bibliographystyle{IEEEtran} 
\bibliography{references}

\end{document}